\documentclass[sigconf,nonacm]{acmart}

\AtBeginDocument{%
  }

\copyrightyear{2025}
\acmYear{2025}
\setcopyright{acmlicensed}
\acmConference[WWW '25] {Proceedings of the ACM Web Conference 2025}{April 28--May 2, 2025}{Sydney, NSW, Australia.}
\acmBooktitle{Proceedings of the ACM Web Conference 2025 (WWW '25), April 28--May 2, 2025, Sydney, NSW, Australia}
\acmISBN{979-8-4007-1274-6/25/04}
\acmDOI{10.1145/3696410.3714919}

\usepackage{graphicx}
\usepackage{multirow}
\usepackage{algorithm}
\usepackage[noend]{algorithmic}
\usepackage{mathtools}
\usepackage{color,url}
\usepackage{calrsfs}
\usepackage[font=normalsize,labelfont={bf}, skip=3pt]{caption}

\usepackage[usenames,dvipsnames]{pstricks}
\usepackage{pstricks-add}
\usepackage{epsfig}
\usepackage{pst-grad} 
\usepackage{pst-plot} 
\usepackage[space]{grffile} 
\usepackage{etoolbox} 
\makeatletter 
\patchcmd\Gread@eps{\@inputcheck#1 }{\@inputcheck"#1"\relax}{}{}
\makeatother

\newcommand{\ts}{\texttt{S}}
\newcommand{\twt}{\texttt{WT}}
\newcommand{\tw}{\texttt{W}}
\newcommand{\td}{\texttt{D}}
\renewcommand{\tt}{\texttt{T}}
\newcommand{\tc}{\texttt{C}}
\newcommand{\tb}{\texttt{B}}

\newcommand{\T}{\mathcal{T}}

\newcommand{\G}{\mathcal{G}}
\newcommand{\V}{\mathcal{V}}
\newcommand{\E}{\mathcal{E}}

\newcommand{\mfd}{\texttt{MajorFlowDetector}{}}

\newcommand{\Gu}{{G_{\cup}}}
\newcommand{\Vu}{{V_{\cup}}}
\newcommand{\Tu}{{T_{\cup}}}

\theoremstyle{definition}

\newcommand{\rev}[1]{{\color{black}{#1}}}
\newcommand{\et}{{\it et al.}}

\begin{document}
\title[Serial Scammers and Attack of the Clones]{Serial Scammers and Attack of the Clones: How Scammers Coordinate Multiple Rug Pulls on Decentralized Exchanges}


%


\author[P. D. Huynh]{Phuong Duy Huynh}
\affiliation{%
  \institution{RMIT University}
  \country{}}
\email{phuong.duy.huynh@rmit.edu.au}

\author[S. H. Dau]{Son Hoang Dau}
\affiliation{
  \institution{RMIT University}
  \country{}}
\email{sonhoang.dau@rmit.edu.au}

\author[N. Huppert]{Nicholas Huppert}
\affiliation{
  \institution{RMIT University}
  \country{}}
\email{nicholas.huppert@rmit.edu.au}

\author[J. Cervenjak]{Joshua Cervenjak}
\affiliation{
  \institution{RMIT University}
  \country{}}
\email{joshua.cervenjak@rmit.edu.au}

\author[H. Sun]{Hoonie Sun}
\affiliation{%
  \institution{RMIT University}
  \country{}}
\email{hoonie.sun@rmit.edu.au}

\author[H. Y. Tran]{Hong Yen Tran}
\affiliation{
  \institution{UNSW Canberra}
  \country{}}
\email{hongyen.tran@unsw.edu.au}

\author[X. Li]{Xiaodong Li}
\affiliation{
  \institution{RMIT University}
  \country{}}
\email{xiaodong.li@rmit.edu.au}

\author[E. Viterbo]{Emanuele Viterbo}
\affiliation{
  \institution{Monash University}
  \country{}}
\email{emanuele.viterbo@monash.edu}

\begin{abstract}
We explored the ubiquitous phenomenon of \textit{serial scammers}, each of whom deployed dozens to thousands of addresses to conduct a series of \textit{similar} Rug Pulls on popular decentralized exchanges.
We first constructed two datasets of around 384,000 scammer addresses behind all one-day Simple Rug Pulls on Uniswap (Ethereum) and Pancakeswap (BSC), and identified distinctive \textit{scam patterns} including \textit{star}, \textit{chain}, and \textit{major (scam-funding) flow}. 
These patterns, which collectively 
cover about $40\%$ of all scammer addresses in our datasets, reveal typical ways scammers run multiple Rug Pulls and organize the money flow among different addresses.
We then studied the more general concept of \textit{scam cluster}, which comprises scammer addresses linked together via direct ETH/BNB transfers or behind the same scam pools. 
We found that scam token contracts are highly similar within each cluster (average similarities $>70\%$) and dissimilar across different clusters (average similarities $<30\%$), corroborating our view that each cluster belongs to the same scammer/scam organization. 
Lastly, we analyze the scam profit of individual scam pools and clusters, employing a novel \textit{cluster-aware profit formula} that takes into account the important role of \textit{wash traders}.
The analysis shows that the existing formula inflates the profit by at least 32\% on Uniswap and 24\% on Pancakeswap.
\end{abstract}

\maketitle 

\vspace{-10pt}
\section{Introduction}
\label{sec:intro}

The total crypto scam revenue from 2019 to 2023 reached a staggering amount of nearly US \$40 billion, according to the latest Crypto Crime Report by the leading blockchain analytics firm Chanalysis~\cite[p.~104]{chainalysis_report_2024}. The report also shows that \textit{Rug Pull}, a common type of scam in the decentralized finance (DeFi) ecosystem, was among the top three fastest growing scams in 2023~\cite[p.~105]{chainalysis_report_2024}. Rug Pull, first reported in 2021~\cite{Xia_etal_2021, chainalysis_report_2022}, refers to scams in which the developer(s) of a cryptocurrency project (usually a new token) suddenly vanished with investors' fund, leaving their purchased assets worthless.
Rug-Pull scams were responsible for the loss of more than US \$100 million in 2023 alone according to Immunefi's Crypto Loss Report~\cite{immunefi_hack_fraud_reports}, and are still costing millions of dollars every month in 2024.
Immunefi's reports~\cite{immunefi_hack_fraud_reports} also identified Etherem and BSC as the two most targeted chains by hacks and Rug Pulls in 2023-2024.


In this work, we investigate how \textit{serial} Rug-Pull scammers operate on UniswapV2 (UNI) and PancakeswapV2 (CAKE), the two most popular decentralized exchanges (DEXs) on Ethereum and BSC. Serial scammers, which could be individual scammers or scam organizations, are those that perform a series of similar scams via multiple scam addresses. The existence of such serial scammers have been sporadically discussed in the literature of Rug Pulls albeit under simplistic definitions.
More specifically, it was reported in the earlier work of Xia~\et~\cite[Sect.~4.4]{Xia_etal_2021} and also in the recent work of Cernera~\et~\cite{Cernera_etal_USENIX2023} that there were addresses that were behind multiple Rug-Pull scams on Uniswap and BSC. This is the simplest form of serial scammers, in which \textit{one} address created multiple scam tokens or pools. It was assumed in~\cite{Cernera_etal_USENIX2023} that scammer addresses are independent, and the Rug Pulls carried out by them are unrelated, leaving the case of single scammers coordinating multiple scams and scam addresses for future research (see~\cite[Sect.~11]{Cernera_etal_USENIX2023}). Xia~\et~\cite{Xia_etal_2021} also studied a related concept of colluding addresses that are behind the same scam pool, which is not the same as serial scammer as the focus is on individual scam pools.

Similar to~\cite{Xia_etal_2021} and \cite{Cernera_etal_USENIX2023}, other existing works on Rug Pulls~\cite{Mazorra_etal_2022, Huynh_etal_arxiv_2023, Nguyen_etal_2023, Zhou_etal_ICSE_SEIP_2024} for ERC-20/BEP-20 tokens on Ethereum and BSC only investigated snapshots of the Rug-Pull scam landscape by \textit{zooming in} to either individual scammer addresses or individual scam tokens/pools and treating them as independent entities. 

\begin{figure}[tb]
    \centering
    \includegraphics[width=0.45\textwidth]{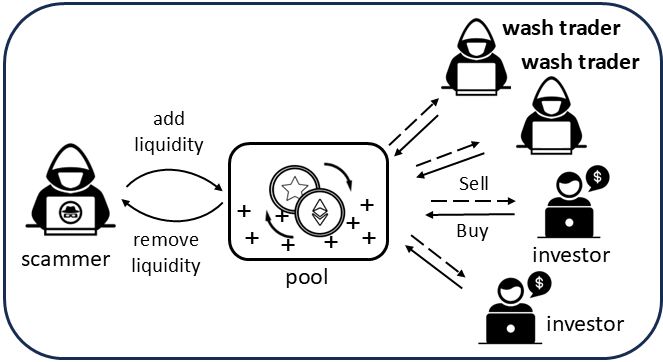}
    \caption{\textmd{Typical activities in a DEX scam pool. There can be one or more \textit{scammer} addresses behind each pool. \textit{Wash traders} bought scam tokens to increase its price and generate fake activities. The arrow refers to the flow of the native token (ETH/BNB). If the scam token is a Trapdoor~\cite{Huynh_etal_arxiv_2023} then it's possible to buy but impossible to sell the scam token to obtain the high-value native token.}} 
    \label{fig:typical_pool}
    \vspace{-10pt}
\end{figure}

Unlike previous approaches, we \textit{zoomed out} 
to reconstruct a more comprehensive 
picture of how serial Rug-Pull scammers organized their operations. 
To facilitate the exposition of serial scammers, we restricted our investigation to \textit{one-day Simple Rug Pull} tokens, which are easy to identify and prevalent on Ethereum and BSC. Such tokens lived for only one day and were paired with a high-value token such as ETH or BNB in a pool, the liquidity of which was provided and removed by a scammer address within a day (see~\cite{Cernera_etal_USENIX2023}). 
As shown in~\cite{Cernera_etal_USENIX2023}, one-day Simple Rug-Pull tokens were abundant on Ethereum and BSC, accounting for nearly $50\%$ of all tokens. 
Note that one-day Simple Rug Pulls are much easier to detect with higher confidence compared to longer-life Rug Pulls, which were usually labelled by less straightforward rules. For example, Mazorra~\et~\cite{Mazorra_etal_2022} labelled as Rug Pulls inactive tokens (no activities for more than a month) where the corresponding pools had their liquidity completely removed or had their price dropped more than 90\% and never recovered. Longer-life Rug-Pull tokens could be potentially mixed up with low-performing tokens where the creator 
removed the liquidity after a long period of no profit without any ill intent.

\textit{Wash trading} is an integral part of Rug Pull scams on DEXs (see Fig.~\ref{fig:typical_pool}), in which fake investors - addresses that are likely owned by the scammers - buy the scam tokens to increase their price, making them look more promising (see, e.g. Xia~\et~\cite[Sec.~5.3.3]{Xia_etal_2021}). 
By first constructing the \textit{scam clusters}, 
we were able to identify all scammer addresses that 
wash traded their own scam pools as well as other pools within the same cluster. 
In other words, these addresses play a double role of scammer and wash trader simultaneously, not only conducting their own Rug Pulls but also buying scam tokens created by other addresses within the same cluster 
to attract the real investors/victims. 
Using the scam clusters allows us to estimate more accurately the \textit{real profits} of the scammers.
As an example, the creator of the Uniswap pool that pairs ETH and a scam token called \texttt{PUMPKIN} added 1.8 ETH to and removed 9.27 ETH from the pool, seemingly reaped a profit of more than 7 ETH (>\$10,000) after mere 37 minutes, creating the illusion of a highly successful scam. However, as discussed in Section~\ref{sec:cluster_aware_scam_profit}, it turns out that all major investors were wash traders, and the \textit{real profit} for the creator of \texttt{PUMPKIN}, 
after deducting the wash-trading expense, is almost zero. 

Our contributions are summarized below.
\begin{itemize}
    \item We constructed two one-day Simple Rug-Pull datasets on Uniswap~\cite{uniswap} and Pancakeswap~\cite{pancakeswap} V2, containing nearly \rev{632,000} scam pools and \rev{384,000} scammer addresses. 
    \item We formally defined and identified typical \textit{scam patterns} in our datasets, including \textit{stars}, \textit{chains}, and \textit{major flows}, \rev{which cover about $40\%$ of all scammer addresses}. The longest scam chain and the largest scam star consist of \rev{713} and \rev{585} scammer addresses, respectively. 
    \rev{The largest major flow has 820 scammer addresses.}
    Such patterns revealed distinctive ways serial scammers coordinated multiple scams on DEXs.
    \item We formalized the concept of \textit{scam clusters}, which consist of scammer addresses that interact with each other via direct ETH/BNB transfers or via scam pool activities. 
    We identified \rev{8,467} and \rev{18,042} clusters, with the largest ones consisting of \rev{4,155} and \rev{23,399} scammer addresses on Uniswap and Pancakeswap, respectively.
    We also found that token contracts used within each cluster are mostly similar, 
    indicating that they could be clones from the same scammer/organization.
    \item We proposed a \textit{cluster-aware} scam profit formula that captures more accurately the scam profit of Rug Pulls on DEXs by taking into account the wash-trading transactions coming from the cluster addresses. Our analysis 
    shows that the existing formula (ignoring wash-trading) has inflated the true scam profit by at least $24\%$ (CAKE) and 32\% (UNI).
\end{itemize}


The paper is organized as follows. Section~\ref{sec:background} provides the necessary background. Section~\ref{sec:data_collection} describes how the Rug-Pull datasets on UniswapV2 and PancakeswapV2 were constructed. Section~\ref{sec:pattern_detection} is devoted to scam pattern detection and analysis. We study the scam clusters and cluster-aware scam profit in Sections~\ref{sec:clustering} and~\ref{sec:cluster_aware_scam_profit}. 
The paper is concluded in Section~\ref{sec:conclusion}.

\section{Background}
\label{sec:background}

\subsection{Ethereum and BNB Smart Chain (BSC)}
\textbf{Ethereum} is the second-most popular blockchain after Bitcoin, with a market capitalization of over US\$300 billion at the time of this study~\cite{ethereum_cap}. 
Representing the second generation of blockchain, it provides EVM, a Turing-complete virtual machine, and natively supports smart-contract-based applications.
Smart contracts are executable pieces of code holding business logic~\cite{Szabo_1997}. Since the integration of smart contracts into blockchain by Ethereum in 2015, the blockchain technology has rapidly progressed and been widely applied to many domains, including supply chain, health care, and governance~\cite{jaiman2020consent,casado2018blockchain,novo2018blockchain}. 
\textbf{BNB Smart Chain (BSC)} was a hard fork from Ethereum in 2020 with a different consensus mechanism to achieve faster transaction processing times with lower fees.

\textbf{Fungible tokens} are interchangeable digital assets governed by smart contracts.
Fungible tokens on Ethereum and BSC must follow the technical standards ERC-20 and BEP-20~\cite{erc20_standard,bep20_standard} by implementing a set of functions and events (see Table~\ref{tab:token_standard} in Appendix~\ref{app:functions_events}). 
\textbf{An account} is a basic unit in blockchains and represented by a unique \textit{address}, which has the prefix ``0x'' followed by a sequence of 40 hexadecimal characters. Ethereum and BSC accounts are classified into 
externally owned account (EOA) and contract account (CA). 
The former is controlled by users via a private-public key pair 
while the latter is managed by a contract that contains executable code.
\textbf{A transaction} in blockchains is a message between two accounts. Transactions record all activities on blockchains, such as deploying a new smart contract or transferring a digital asset. A fee will be charged to the user when creating a transaction to compensate the miner who processes it. 
Transactions are classified based on the type of the sender account: a \textit{normal} transaction is sent from an EOA, while an \textit{internal} transaction is sent from a CA.

\subsection{Uniswap (UNI) and Pancakeswap (CAKE)}
Decentralised exchanges (DEXs) are financial platforms that operate based on price determination mechanisms such as order books and automated market makers (AMM), allowing users to exchange their digital assets without the involvement of central authorities~\cite{werner2022sok}. Uniswap~\cite{uniswap}, which debuted on Ethereum in 2018, is one of the most popular DEXs. It 
is the first DEX that adopted the AMM mechanism successfully with the concept of \textit{exchange pools} and \textit{liquidity providers}. 
An exchange pool in Uniswap operates like a money-exchange counter of two currencies (fungible tokens).
The liquidity in a pool is provided by one or multiple users (liquidity providers). Anytime a provider adds liquidity (in the form of two corresponding tokens) into a pool, the pool's smart contract will ``mint'' LP tokens and send them to the provider as liquidity shares. To withdraw, the provider simply sends LP tokens back to the pool, which will ``burn'' these tokens and return the funds to the provider. 

Although Uniswap has launched its fourth version (February 2025), other versions are still operating as independent platforms. Among them, Uniswap version 2 (UniswapV2) entirely outperforms others in terms of the number of listed tokens and pools~\cite{coingecko_dex}. Due to the popularity of this version and its open-source smart contracts, more than 650 DEXs across different blockchains are the forks of UniswapV2~\cite{uniswap_forks}, and  PancakeswapV2~\cite{pancakeswap} is the most successful fork on BSC. As such, we chose to study UniswapV2 and PancakeswapV2, noting that our approach is also applicable to other forks of UniswapV2 and other similar DEXs.

\section{One-Day Simple Rug-Pull Datasets} 
\label{sec:data_collection}

We first define one-day Simple Rug Pull scam and then discuss how to construct the datasets of scam pools, scam tokens, and scammer addresses on UniswapV2 (Ethereum) and PancakeswapV2 (BSC).
The definition of scammer addresses were first discussed in Xia~\et~\cite[Sec.~5.3]{Xia_etal_2021}. 
See Fig.~\ref{fig:scammer_history} for an example of the transaction history of a typical Rug-Pull scammer address. 

\begin{figure}[htb]
    \centering
    \includegraphics[scale=0.25]{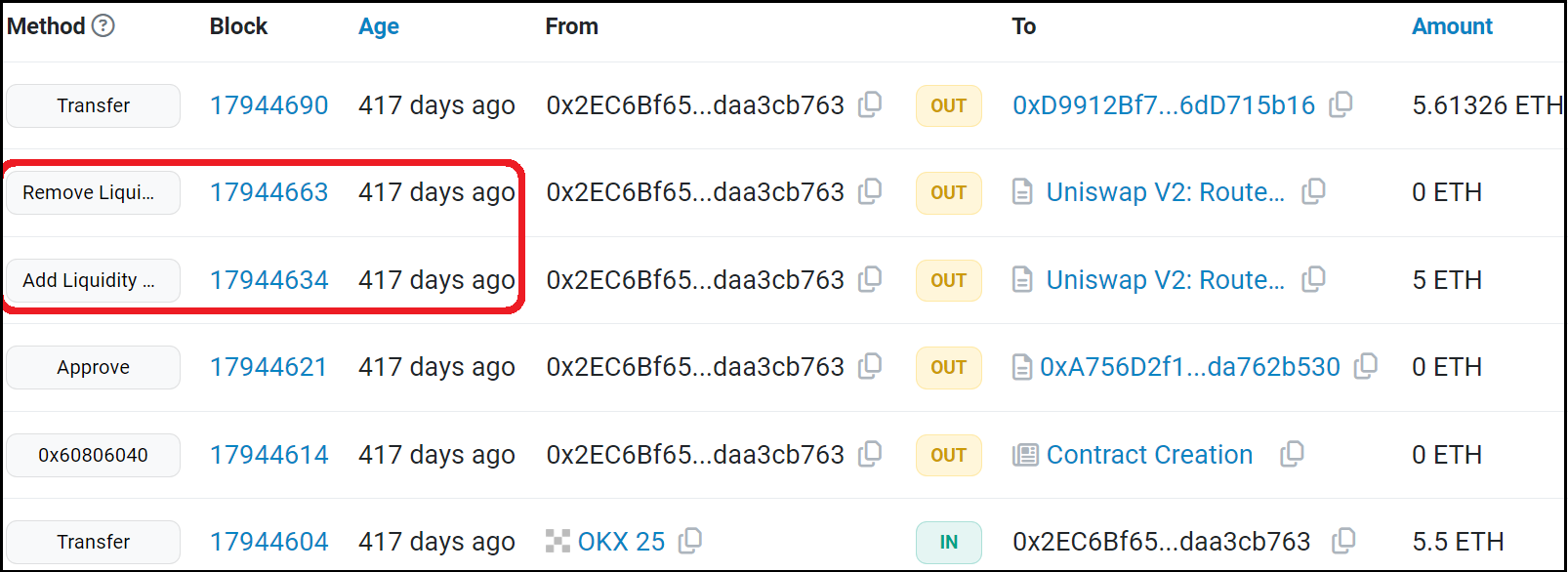}
    \caption[abc]{\textmd{The transaction history of a typical one-day Simple Rug-Pull scammer address on Etherscan, with one ``Add Liquidity'' and one ``Remove Liquidity'' events within a day. The scammer address 0x2ec6bf65bf9cf83bdd9295425b5b145daa3cb763 received 5.5 ETH from the public exchange OKX, created a pool on UniswapV2, added 5 ETH and 250M LIGHT (scam token) as liquidity, then removed liquidity (5.19 ETH and 241M LIGHT) within a day.}}
    \label{fig:scammer_history}
\end{figure}

\vspace{-5pt}
\subsection{One-Day Simple Rug Pull}
\label{subsec:one_day_rp}

We follow the definition of one-day Rug Pull in Cernera~{\et}~\cite[Sec.~7.1]{Cernera_etal_USENIX2023}.
\rev{It is required that the scammer burned at least 99\% only instead of 100\% of liquidity tokens to account for rounding error.} 

\begin{definition}[One-Day Simple Rug Pull]
\label{def:rugpull}
    A \textit{one-day exchange pool} is a pool that pairs a higher-value (lower reserve) token and a lower-value (higher reserve) token in which the first and last events happen within a day. 
     An exchange pool is called a \textit{one-day Simple Rug Pull} if  a) it is one-day, b) the lower-value token is paired in this pool only, and c) it has one \texttt{Mint} and 
    one \texttt{Burn} event that burns at least 99\% of the minted LP tokens. 
    The corresponding lower-value token of a one-day Simple Rug-Pull pool is called a one-day Simple Rug-Pull token. We also refer to them as \textit{scam pool} and \textit{scam token}. 
\end{definition}

\begin{definition}[Scammer Addresses]
    \label{def:scammer_address}
    Given a scam pool, we refer to the addresses of the scam token creator, the scam pool creator, the liquidity provider, and the liquidity remover as \textit{scammer addresses} behind/associated with the pool. 
\end{definition}

\subsection{Data Collection}
We collected data on UniswapV2 and PancakeswapV2 from the time they were launched to the time of this study (July 2024).

\textbf{Exchange Pools.}
Uniswap works based on three main contracts: \texttt{Factory}, \texttt{Pair}, and \texttt{Router}.
\texttt{Factory} generates an \textit{exchange pool} (a \texttt{Pair}) for users from two given tokens. 
\texttt{Factory} stores the addresses of created \texttt{Pairs} while a \texttt{Pair} stores the addresses of two listed tokens.
We first used Web3.py~\cite{web3} to query all created pools by calling the \texttt{allPairs} function of \texttt{Factory}. For each collected pool, we then called the functions \texttt{token0} and \texttt{token1} from its contract to retrieve the pair of listed tokens. Finally, we used Etherescan and BSCscan APIs~\cite{EtherscanAPI,BSCscanAPI} for retrieving useful information of such pools and tokens, including creator addresses and contract source codes. As a result, we collected 356,295 pools and 343,637 tokens on Uniswap, and 1,694,058 pools and 1,510,774 tokens on Pancakeswap.

\textbf{Pool Events.} Pool events are the logs of a \texttt{Pair} written down when the state of any property changes. In our study, we collected four pool events, including \texttt{Mint}, \texttt{Burn}, \texttt{Transfer}, and \texttt{Swap}. The first three events occur every time LP-token information is changed. For example, an exchange pool emits a \texttt{Mint} event each time LP-tokens are minted for a new liquidity adding. Similarly, a \texttt{Burn} event is emitted when a pool burns the LP tokens received from a liquidity \rev{remover}. A \texttt{Transfer} event is recorded any time an LP token is transferred from one address to another (ownership change). On the other hand, 
a \texttt{Swap} event is written each time a user swaps tokens in an exchange pool.
To collect these events, we used the ``getLogs" APIs from Etherscan and BSCscan. To reduce the time and computation cost, we only downloaded the first 1000 events of each pool, as it will not impact our goal of collecting one-day scam pools/tokens (which should not have too many events). The downloaded data was then decoded to extract useful 
information using our event decoder. 
As the result, we collected 2.2 million \texttt{Mint} events, 1.0 million \texttt{Burn} events, 4.7 million \texttt{Transfer} events and 49.9 million \texttt{Swap} events on Uniswap (resp. 17.3 million, 5.1 million, 38.0 million, and 154.9 million events on Pancakeswap).


\textbf{Rug Pulls and Scammers.}
From the pool events gathered in the previous step, we identified Rug Pull scams on each DEX following 
the one-day Rug Pull definition. To that end, we collected all pools that only fired one \texttt{Mint} and one \texttt{Burn} event and examined if 99\% of liquidity was burned within a day of it being added. Moreover, we only focused on ETH pools on Uniswap and BNB pools on BSC. This allows us to accurately assess the costs and profits associated with these scams. According to our analysis, 96\% 
of pools on Uniswap are ETH pools and  88\% 
of pools on Pancakeswap are BNB pools. Thus, the missing cases are minority 
and were ignored. 
\rev{Finally, we extracted scammer addresses, who created a scam pool or token, or provided/withdrew liquidity 
from identified scam pools (see Def.~\ref{def:scammer_address})}. 
Note that we excluded all public/service addresses (e.g., CEX, DEX, bots, bridges, mixers) and old addresses that had no transactions within our data collection period. 
In the end, 
161,329 (45\%) scam pools and 145,654 unique scammers were found 
on Uniswap, and \rev{470,712 (28\%)} 
scam pools and \rev{238,280} unique scammers were found on Pancakeswap. \rev{See App.~\ref{app:data_collection} for further analyses.}


\section{Serial Scam Patterns}
\label{sec:pattern_detection}

We explore in this section distinctive \textit{funding patterns} that reveal how scammer addresses, which potentially belong to the same serial scammers, receive fund to carry out their Rug-Pull scams and transfer the fund to another address in a \textit{highly coordinated} manner. 
Moreover, such patterns collectively cover significant portions of scammer addresses: 37.1\% on Uniswap and 44.3\% on Pancakeswap.
This is the first step in our investigation of serial scammers on DEXs where we deviate from most existing approaches~\cite{Xia_etal_2021, Cernera_etal_USENIX2023,Mazorra_etal_2022, Huynh_etal_arxiv_2023, Nguyen_etal_2023, Zhou_etal_ICSE_SEIP_2024}, which often treat different scams as unrelated. 

We define the following concepts to facilitate our discussion. 
Given an address A, an \textit{in-transaction} is a native-token (ETH/BNB) transfer from another address B to A, and B is called an \textit{in-neighbor} of A. Similarly, an \textit{out-transaction} is a native-token  transfer from A to another address C, referred to as an \textit{out-neighbor} of A. We refer to the buys and sells of scam tokens as \textit{swap-ins} (paying ETH/BNB to the pool) and \textit{swap-outs} (receiving ETH/BNB from the pool).

\subsection{Scam Stars}
\label{subsec:star}

We start our scam pattern exploration by defining three types of scam star, representing a commonly found pattern in which scammer addresses are coordinated by a center address. 

\begin{definition}[Scam Star]
    \label{def:star}
    A \textit{scam star} consists of a \textit{center} address $c$ (coordinator) and $n\geq 5$ \textit{scammer} addresses $s_1,\ldots,s_n$ (satellites) that satisfy one of the following patterns.
    
    \textbf{OUT-star (common funder)}: 
    $s_1,\ldots,s_n$ received at least 100\% of the cost to create their first scam from $c$ but sent no fund back to $c$. The corresponding transfer from $c$ must be the largest in-transaction each satellite received \textit{before} conducting the first scam.  
    
    \textbf{IN-star (common beneficiary)}: 
    $s_1,\ldots,s_n$ received no fund from $c$, but transferred at least $p=90\%$ of their last scam revenue back to $c$. The corresponding transfer to $c$ must be the largest out-transaction from each scammer \textit{after} conducting their last scam. 
    
    \textbf{IN/OUT-star (common funder/beneficiary)}: the satellites 
    $s_1,\ldots,s_n$ received at least 100\% of the cost to create their first scams from $c$. The corresponding in-transaction from $c$ must be the largest in-transaction each satellite received \textit{before} conducting the first scam. Moreover, the satellites transferred at least $p=90\%$ of their last scam's revenue back to $c$. The corresponding out-transaction to the center must also be the largest out-transaction from each scammer \textit{after} conducting the last scam. We also use I/O for short. 
\end{definition}

\begin{figure}[tb]
    \centering
    \includegraphics[scale=1]{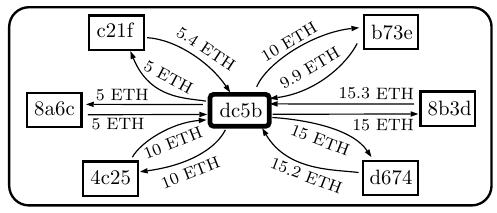}
    \caption[abc]{\textmd{Examples of an IN/OUT-star with \textbf{dc5b} as the center. 
    The six satellites are all scammer addresses. The center's full address is 0xbfc6cc4676aef7216e597d45d68463097520\textbf{dc5b}.}}
    \label{fig:stars}
\end{figure}

Examples of an IN-star and an IN/OUT-star are given in Fig.~\ref{fig:stars}. Note that while we require that the funding amount from the center must cover 100\% of the scam cost (in OUT-stars and an IN/OUT-stars), the value of the OUT transaction only needs to cover at least $p=90\%$ of the last scam's revenue. This reflects our observation that a scammer address may also operate as a wash trader, who spends some small portion of the sum received from the funder and from its scam pool to buy (i.e. \textit{wash trade}) scam tokens from scam pools created by \textit{other} scammer addresses.

\begin{table}[htb]
    \centering
    \setlength{\tabcolsep}{2pt}
    \begin{tabular}{|l|l|l|l|l|l|l|}
    \hline
         Type & \#Stars&     Size &      Fund In &   Fund Out &      Period & \#Scams \\
         \hline
         IN &   1,575  &  585;19&   17597; 104&          &      925;96 & 585;20\\
        \hline
         OUT &  61 &   66;10&       &           19,646;465&    476;56 & 66;11\\
        \hline
         I/O & 73 &   159;15&  22,447;557& 22,734;540&    476;56 & 159;16\\
         \hline
    \end{tabular}
    \vspace{5pt}
    \caption{\textmd{Statistics for maximal \textit{scam stars} found in our \textbf{Uniswap} scammer dataset. Rounded maximum and average values are reported, e.g., the maximum and average size of an IN scam star is 585 and 19 scammer addresses, respectively. We set $p=90\%$.}}
    \label{tab:star_stats_uniswap}
\end{table}

\begin{table}[htb]
    \centering
    \setlength{\tabcolsep}{0.5pt}
    \begin{tabular}{|l|l|l|l|l|l|l|}
    \hline
         Type &  \#Stars &  Size &      Fund In &   Fund Out &      Period & \#Scams \\
         \hline
         IN &  1,301 &   977;22 & 5,734;79&          &     1,150;146 & 1,008;27\\
         \hline
         OUT & 339 & 1,900;28   &     & 6,513;101        & 1,184;180   & 2,824,44\\
        \hline
         I/O & 251 & 9,012;59 & 939,528;4,293 & 592,168;2,699  & 1,025;51  & 9,183;66\\
        \hline
    \end{tabular}
    \vspace{5pt}
    \caption{\textmd{Statistics (maximum/average values) for maximal \textit{scam stars} found in our \textbf{Pancakeswap} dataset ($p=90\%$).}}
    \label{tab:star_stats_Pancakeswap}
\end{table}

\textbf{Detecting Scam Stars.} Given a list of scammer addresses, we developed \texttt{StarDetector} (see App.~\ref{app:star}), 
an algorithm that detects all scam stars containing such addresses. 
We ran our star detection algorithm on the both scammer datasets and report the statistics in Tables~\ref{tab:star_stats_uniswap} and~\ref{tab:star_stats_Pancakeswap}.
Note that ``\#Stars'' is the number of stars found in our datasets, ``Size'' is the number of scammer addresses in a star, ``Fund In''/''Fund Out'' are the total fund coming in/out of a center, ``Period'' is the number of days between the earliest to the lattest scams, and ``\#Scams'' counts the number of scam pools in a star. Setting $p=80$-$95\%$ yields similar statistics (see App.~\ref{app:star}).

The stars cover 32,597 unique scammer addresses, accounting for 22.4\% of the Uniswap scammer dataset.
For Pancakeswap, the stars cover 51,368 unique scammer addresses, or 21.6\% of the dataset.




\subsection{Max-In-Max-Out Scam Chains}
\label{subsec:chain}

Our goal is to capture the main \textit{funding flow} among scammer addresses. One typical flow shape is a \textit{chain} (or a path), in which each scammer address was funded by another scammer address, carried out scams, and transferred fund to the next scammer address. 

\begin{definition}[Simple Scam Chain]
    \label{def:chain-simple}
    A \textit{simple scam chain}, or a \textit{max-in-max-out scam chain}, is a list of $n$ \textit{scammer} addresses $s_1,s_2,\ldots,s_n$ ($n \geq 2$) satisfying the following conditions.
    \begin{itemize}
        \item[(C1)] $s_i$ is the \textit{largest funder} of $s_{i+1}$, and $s_{i+1}$ is the \textit{largest beneficiary} of $s_i$, for every $i\in \{1,\ldots,n-1\}$.
        \item[(C2)] The transfer from $s_i$ to $s_{i+1}$ occurred \textit{after} $s_i$ has completed its last scam and \textit{before} $s_{i+1}$ started its first scam. 
    \end{itemize}
    A scam chain is \textit{maximal} if no other EOAs can be added to obtain a longer chain. We only consider maximal scam chains in this work. It is obvious that each scammer belongs to at most one chain. 
\end{definition}

\vspace{-10pt}
\begin{figure}[htb]
    \centering
    \includegraphics[scale=0.9]{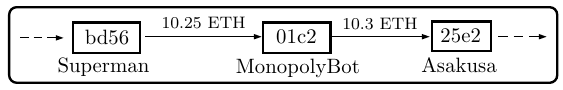}
    \caption[abc]{\textmd{Part of a \textit{simple scam chain} of length 47, starting with 0xc3e8290045952d520f4c2eb7e8725cabc4c8\textbf{b5d6}. Each address performed one Rug Pull and then transferred fund to the next one.}}
    \label{fig:chain_simple}
    \vspace{-5pt}
\end{figure}

We designed a simple algorithm 
to find all maximal scam chains among a given list of scammer addresses. 
Our findings on both datasets are reported in Table~\ref{tab:chain_stats}. 
``\#Chains'' is the number of chains found in the dataset, ``Length'' is the number of scammer addresses in a chain, ``Ave. Transfer'' is the average amount of native tokens transferred between two consecutive addresses, ``Period'' is the number of days between the first and last scams, and ``\#Scams'' is the total number of scam pools in a chain. 
All numbers are rounded.

\begin{table}[htb]
    \centering
    \setlength{\tabcolsep}{1pt}
    \begin{tabular}{|l|l|l|l|l|l|}
    \hline
         \textbf{DEX} & \textbf{\#Chains} &    \textbf{Length} &      \textbf{Ave. Transfer} &    \textbf{Period (days)} & \textbf{\#Scams} \\
         \hline
         UNI & 4,494  &  274;4&      642;35 & 369;3 & 339;6 \\
        \hline
         CAKE & \rev{14,028} & \rev{713;4}&   \rev{1,486;13}  &        \rev{251;2}  & \rev{142;2}\\
        \hline
    \end{tabular}
    \vspace{5pt}
    \caption{\textmd{Statistics for maximal \textit{scam chains} in \textbf{Uniswap} and \textbf{Pancakeswap} scammer datasets. Max/average values are reported.}}
    \label{tab:chain_stats}
\end{table}

Uniswap scam chains cover 19,977 unique scammer addresses, or 13.7\% of the dataset. For Pancakeswap, they cover 52,014 addresses, or 21.8\% of the dataset. See Apps.~\ref{app:scam_chain} and~\ref{app:scam_patterns_more_analysis} for further analyses.

We observe that current scam labelling on Etherscan and the alike is reported case by case perhaps due to \textit{the lack of a cluster view}. For example, 0xcba70cc7ea0ff51222a24f659515b1e425a2a913 was labeled as a scam on Etherscan. This address then transferred 74 ETH to 0x3be0cc5ef23b03da49519bea21187901d95b6875, which conducted another Rug Pull but has not been reported. Both addresses belong to a scam chain of length 23 detected by our algorithm.

\subsection{Major Scam-Funding Flows}
\label{subsec:majority_flow}

While the max-in-max-out scam chains capture many cases where each scammer address had a major funder and transferred most scammed fund to another scammer address, they miss the complex cases where there are \textit{more than one} major funders and/or beneficiaries. 
For example, when two addresses funded the scammer with 10 and 5.2 ETH (as in Fig.~\ref{fig:majority_flow}), only the one funding 10 ETH will be added to the chain. 
Intuitively, in a more general scam-funding flow, each scammer address was first fully funded by one or more scammer addresses (major funders), then carried out one or several scams, and after all the scams are completed, transferred most of the scam revenue to other scammer addresses (major beneficiaries).
For example, in Fig.~\ref{fig:majority_flow}, the chain can only capture addresses along the path \textbf{e3df}$\to$\textbf{5a95}$\to$\textbf{fc34}, hence missing \textbf{9cb0} and \textbf{9dbb}. 

\begin{figure}[htb]
    \centering
    \includegraphics[scale=1]{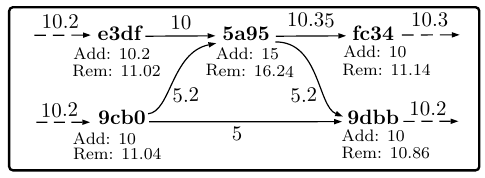}
    \caption[]{\textmd{A minimal \textit{major flow}. The \textit{input} address \textbf{9cb0} (full address: 0x15a828abe5ef29fa9fbe5c0774110232f908\textbf{9cb0}) added 10 ETH into a scam pool and removed 11.04 ETH. It then transferred 5.2 ETH and 5 ETH to its major beneficiaries \textbf{5a95} and \textbf{9dbb}. The \textit{internal} address \textbf{5a95}, after funded by its major funders \textbf{9cb0} and \textbf{e3df}, ran one scam and transferred fund to its major beneficiaries, the \textit{output} addresses \textbf{fc34} and \textbf{9dbb}. Note that \textbf{5a95} was \textit{fully funded} by its major funders and transferred at least 90\% of its scam revenue (obtained from the last pool) to its major beneficiaries.}}
    \label{fig:majority_flow}
\end{figure}

We formalize below the more sophisticated concept of \textit{major scam-funding flow}, or \textit{major flow} for short, which takes into consideration the funding amount required to fund a scam (when providing liquidity) and its revenue (when removing liquidity), and captures the intuition of a general scam-funding flow discussed above.
A major flow consists of three kinds of (scammer) addresses: \textit{input}, \textit{internal}, and \textit{output}. Each \textit{internal} address received 100\% funding from input addresses and/or other internal addresses (\textit{major funders}) to carry out its first scams, then transferred at least $p=90\%$ of its last scam's revenue to other internal and/or output addresses (\textit{major beneficiaries}). We observed that a scammer address may also spend some of its scam revenue on wash trading other scam pools or other activities unrelated to scam funding. Hence, only $p=90\%$ is required instead of 100\% to account for those small spendings (we also examine other values of $p$). \textit{Input} and \textit{output} addresses behaved similarly, but the funders of input  and beneficiaries of output addresses do not belong to the cluster, respectively. 
\begin{definition}[Scam-Funding Transactions]
    \label{def:major_transactions}
    For a scammer address $s$, let $T_F(s)$ be the minimum set of in-transactions to $s$ with the highest values \textit{before the first scam} of $s$ occurred that together provide enough funding for its first scam. Let $T_B(s)$ to be the minimum set of out-transactions from $s$ with the highest values \textit{after the last scam} of $s$ occurred that together cover at least $p=90\%$ of the revenue of its last scam. The sets $F(s) \triangleq \big\{f \colon t=(f,s) \in T_F(s)\big\}$ and $B(s) \triangleq \big\{b\colon t=(s,b) \in T_B(s)\big\}$ are called the sets of \textit{major funders} and \textit{major beneficiaries} of $s$, respectively. 
\end{definition}

\vspace{-5pt}
\begin{definition}[Major Flow]\label{def:majority_flow}
    Let $S$ be a set of scammer addresses. A \textit{major (scam-funding) flow} in $S$ is a directed graph $G=(V,T)$ with the vertex set $V\subseteq S$ ($|V|\geq 2$) representing the scammer addresses and the edge set $T$ representing the native-token transfers between scammer addresses in $V$ that satisfies (P1) and (P2) given below. For $s \in V$, we use $I_G(s)$ and $O_G(s)$ to denote the set of incoming edges to $s$ and outgoing edges from $s$ in $G$, respectively.   
    \begin{itemize}
        \item[(P1)] For every $s\in V$, if $I_G(s)\neq \varnothing$ then $I_G(s) \equiv T_F(s)$, and if $O_G(s) \neq \varnothing$ then $O_G(s) \equiv T_B(s)$. 
        \item[(P2)] $G$ is \textit{weakly} connected, i.e. there is a path from $u$ to $v$ in the underlying \textit{undirected} graph of $G$ for every two distinct vertices $u$ and $v$ in $V$.
    \end{itemize} 
    An address $s \in V$ is an \textit{input} if $I_G(s)=\varnothing$ and $O_G(s)\neq \varnothing$, an \textit{output} if $I_G(s)\neq \varnothing$ and $O_G(s)= \varnothing$, and an \textit{internal} address if $I_G(s)\neq \varnothing$ and $O_G(s)\neq \varnothing$.
    \textit{Minimal} major flows are those containing no proper major flows as subgraph. \textit{Maximal} flows, on the other hand, are not contained as proper subgraph in any other major flows.
\end{definition}

\begin{example}
    \label{ex:majority_flow}
    The major flow $G$ given in Fig.~\ref{fig:majority_flow} consists of five scammer addresses ending with \text{e3df}, \text{9cb0}, \text{5a95}, fc34, and \text{9dbb}.
    We can verify that (P1) are satisfied for all vertices in $G$. For instance, (P1) is satisfied for 5a95 since $I_G(\text{5a95}) = \{\text{e3df},\text{9cb0}\}$, which consists of its two major funders: c80f and f9f8 together transferred $15.2=10+5.2$ ETH, which covers completely the amount of 15 ETH that 5a95 added to its scam pool. Also, $O_G(\text{5a95})=\{\text{fc34},\text{9dbb}\}$, which consists of the two major beneficiaries of 5a95: after removing 16.24 ETH from its last scam pool, it transferred $15.55 = 10.35+5.2$ ETH, which is $\geq 90\%$ of 16.24 ETH, to fc34 and 9dbb. (P2) is satisfied as $G$ is weakly connected. 
    Note that $G$ is a \textit{minimal} major flow containing $t_1 \triangleq (\text{e3df},\text{5a95})$. Indeed, as both $t_1$ and $t_2=(\text{9cb0},\text{5a95})$ belong to $T_F(\text{5a95})$ and $t_1 \in G$, according to (P1), $t_2\in G$. Next, since both $t_2$ and $t_3 \triangleq (\text{9cb0},\text{9dbb})$ belong to $T_B(\text{9cb0})$ and $t_2 \in G$, (P1) implies that $t_3\in G$. Similarly, as both $t_3$ and $t_4 \triangleq (\text{5a95},\text{9dbb})$ belong to $T_F(\text{9dbb})$, (P1) requires that $t_4\in G$. Finally, as $t_4$ and $t_5 \triangleq (\text{5a95},\text{fc34})$ belong to $T_B(\text{5a95})$, (P1) implies that $t_5\in G$. 
\end{example}

Finding all (maximal) major flows among a given list $S$ of scammer addresses is a nontrivial task due to the strong properties required by the clusters. 
In particular, (P1) requires that all major funders/beneficiaries of every address must be either all present or all absent in the cluster. It also implies that for every edge/transfer $t=(f,b)$ in the cluster, $f$ must be a major funder of $b$ and $b$ must also be a major beneficiary of $f$ (see Def.~\ref{def:major_transactions}).   

We proposed \mfd{} (see App.~\ref{app:majority_flow}), an efficient algorithm that identifies all (maximal) major flows among all scammers in $S$ in polynomial time in $|S|$ and $|T(S)|$ - the total number of transactions within $S$.
The key idea is to first identify all \textit{minimal} major flows, i.e., those that satisfy (P1)-(P2) and contain no proper major-flow sub-clusters. These minimal clusters will then be merged (if they share common addresses) to form \textit{maximal} clusters.

Theorem~\ref{thm:majority_flow} captures the correctness and complexity of the algorithm \mfd. Its proof can be found in App.~\ref{app:majority_flow}. 

\begin{theorem}[Major Scam-Funding Flow]
    \label{thm:majority_flow}
    Let $S$ be a set of scammer addresses. Then the following statements hold.
    \begin{itemize}
        \item[(a)] The union of two major flows that have at least one vertex in common is always a major flow. 
        \item[(b)] Step~2 of \mfd{} returns all the \textit{minimal} major flows $T_1,\ldots,T_m$ in $S$.
        \item[(c)] Step~3 of \mfd{} returns all the \textit{maximal} major flows in $S$ in time $O\big(|S||T(S)|^2)\big)$, where $T(S)$ denotes the set of all transactions to or from addresses in $S$. 
    \end{itemize}
\end{theorem}

\textbf{Identifying Maximal Major Flows.} Although the major flows require very strict properties, we were still able to find a good number of them on both chains (see Table~\ref{tab:majority_flow_uniswap_pancake}, and also Tables~\ref{tab:majority_flow_uniswap} and~\ref{tab:majority_flow_pancake} in App.~\ref{app:majority_flow}). 
The ``Size'' is the number of scammer addresses in a cluster, the ``Width'' is the \textit{maximum} size of a \textit{minimal} major flow contained in a \textit{maximal} major flow. For example, the major flow given in Fig.~\ref{fig:majority_flow} has size 5 and width 5 as it is a minimal one. 
Note that those with \textit{widths} two are special max-in-max-out chains that also take into account the relation between the transfer amount and the scam cost/revenue.
Those with widths larger than two represent the more sophisticated clusters with some nodes having more than one major funder/beneficiary.
``Fund In'' refers to the total incoming fund from the major funders of all \textit{input} addresses. ``Fund Out'' refers to the total outgoing fund to the major beneficiaries of all \textit{output} addresses. For example, the major flow given in Fig.~\ref{fig:majority_flow} has fund-in $20.4=10.2+10.2$ ETH and fund-out $20.5 = 10.3 + 10.2$ ETH.

\begin{table}[htb]
\setlength{\tabcolsep}{4pt}
\begin{tabular}{|l|l|l|l|l|l|}
\hline
\textbf{DEX} & \textbf{\#Cluster} & \textbf{Size} & \textbf{Width} & \textbf{Fund In} & \textbf{Fund Out} \\ \hline
 UNI                            & 5,298                                    & 156;4      & 7;3       & 1,645;32     & 1,797;34      \\ \hline
 CAKE                             & 16,467                                   & 820;4      & 6;3       & 1,099;12     & 2,626;12     \\ \hline
\end{tabular}
\vspace{5pt}
\caption{\textmd{Statistics (maximum; average values) for \textit{maximal major flows} in \textbf{Uniswap} and
\textbf{Pancakeswap} scammer datasets ($p=90\%$).}} 
\label{tab:majority_flow_uniswap_pancake}
\end{table}

Major flows cover 18,692 unique scammer addresses (12.8\%) on Uniswap, and 58,309 unique scammer addresses (24.5\%) on Pancakeswap. 
All three patterns (chain, star, major flow) cover 54,054 unique scammer addresses, corresponding to 37.1\% of all scammer addresses in the Uniswap scammer dataset. For Pancakeswap, these patterns cover 105,660 unique scammer addresses, corresponding to 44.3\% of the Pancakeswap dataset. See App.~\ref{app:scam_patterns_more_analysis} for more analysis.

\section{Scam Clusters}
\label{sec:clustering}

We have seen in Section~\ref{sec:pattern_detection} how around 40\% of scammer addresses operated following special scam patterns. 
In this section, we investigate more general \textit{scam clusters}, each of which consists of scammer addresses linked together via direct native-token (ETH/BNB) transfers or via scam pools. 
Note that all scam patterns investigated in Section~\ref{sec:pattern_detection} make use of direct native-token transfers only. 

To corroborate our view that addresses in the same cluster are likely owned by the same serial scammer, we developed \textit{AST-Jaccard score} (see App.~\ref{app:similarity_score}), a new code-similarity score based on a careful integration of contract source code's \textit{abstract syntax tree} (AST), \textit{hash function}, and \textit{Jaccard similarity} to overcome typical code obfuscation techniques used in contract cloning. We measured the \textit{similarities} of scam contracts, and found that \textit{intra-cluster} contracts are mostly similar (average similarity score greater than 70\%) while \textit{inter-cluster} contracts are mostly dissimilar (average similarity score less than 30\%). This is yet another strong indicator that \textit{scammer addresses within each cluster are controlled by the same serial scammer}, apart from the fact that such addresses are already either behind the same scam pools or tightly connected in the transfer network. 

\subsection{Generating Scam Clusters}

\begin{definition}[Scam Cluster]
    \label{def:scam_cluster} 
    Given a set $S$ of all scammer addresses, a \textit{scam cluster} is an undirected graph $C=\big(V(C),E(C)\big)$, where $V(C)\subseteq S$ has at least two addresses and $E(C)\subseteq V(C)\times V(C)$ is a set of edges among them, satisfying the following conditions. 
    \begin{itemize}
        \item (C1) An edge $e=(u,v) \in E(C)$ exists if and only if $u$ and $v$ had a direct native-token transfer (ETH/BNB) on the blockchain, or $u$ and $v$ are different scammer addresses associated with the same scam pool (see Def.~\ref{def:scammer_address}). 
        \item (C2) $C$ is connected, that is, for every $u,v\in V(C)$, $u\neq v$, there exists a path from $u$ to $v$ in $C$.
    \end{itemize}
    Note that we only consider \textit{maximal} scam clusters, i.e., no new scammer address can be added to achieve a larger one. 
\end{definition}

Given a dataset $S$ of all scammer addresses, one can identify all (maximal) scam clusters by first forming an undirected graph $G=(V,E)$ with $V=S$ and $E$ is formed using (C1), replacing $V(C)$ and $E(C)$ by $V$ and $E$, respectively, and then find all of its \textit{connected components}. 
The scammer addresses within each connected component form a scam cluster. We ran this simple algorithm on both the Uniswap and Pancakeswap datasets of scammer addresses and identified \textit{all} one-day Simple Rug Pull scam clusters on these DEXs. 

We found 8,467 clusters on Uniswap and 18,042 groups on Pancakeswap, which are formed from 44,650 (30.7\%) and 120,101 unique scammer addresses (50.4\%), respectively. The biggest Uniswap cluster contains 4,155 scammer addresses, while the biggest cluster on Pancakeswap contains 23,399 unique scammer addresses.

\begin{figure}[htb]
    \centering
    \includegraphics[width=0.4\textwidth]{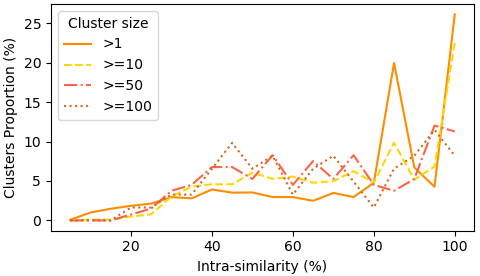}
    \caption{\textmd{\textit{Intra-cluster} similarities for Pancakeswap scam clusters.}}
    \label{fig:intra_pan}
    \vspace{-12pt}
\end{figure}

\begin{figure}[htb!]
    \centering
    \includegraphics[width=0.4\textwidth]{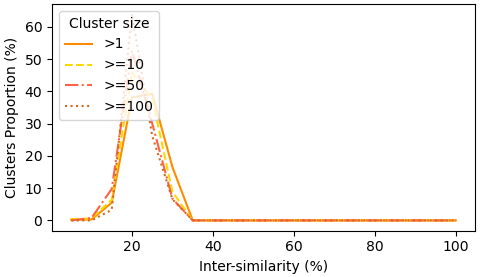}
    \caption{\textmd{\textit{Inter-cluster} similarities for Pancakeswap scam clusters.}}
    \label{fig:inter_pan}
\end{figure}

\subsection{Measuring Cluster Contract Similarity}
\label{subsec:cluster_contract_similarity}


We first retrieved the contract source codes for the scam tokens in our datasets from Etherscan~\cite{EtherscanAPI} and BSCScan~\cite{BSCscanAPI}. 
It is known that nearly 99\% of contracts on Ethereum and 91\% of contracts on BSC do not make their source codes public~\cite{eth_verified_contract,bsc_verified_contract}. 
We found that only 5,080 clusters on Uniswap and 12,495 clusters on Pancakeswap have at least two verified tokens.
We only examined these clusters.

\textbf{Intra-cluster similarity.} 
For each cluster, 
we computed all pairwise similarity scores and took the average. 
Noticeably, 1,390 clusters on Uniswap (27.4\%) and 3,069 clusters on Pancakeswap (24.6\%) have intra-cluster similarity 100\%. More than 50\% clusters on both DEX have intra-cluster similarity \textit{at least} 80\%. On average, intra-cluster similarities are 74\% on Uniswap and 73\% on Pancakeswap. 


\textbf{Inter-cluster similarity.}
Due to high computational cost, we randomly selected up to 100 tokens in each cluster to compare with other clusters. We also selected randomly 500 clusters to pair with the chosen cluster. We repeated this ten times and took the average. 
We found that 75\% and 99.97\% scam clusters on Uniswap and Pancakeswap have inter-cluster similarity scores \textit{at most} 30\%, respectively. Moreover, the average pairwise similarity score is 27\% on Uniswap and 21\% on Pancakeswap. 
Figs.~\ref{fig:intra_pan} and~\ref{fig:inter_pan} shows the intra-/inter-cluster similarity score distributions for the Pancake clusters. 
More analyses can be found in App.~\ref{app:scam_clusters}.
The observed low inter-cluster similarity scores and high intra-cluster similarity scores reinforce our view that \textit{each cluster corresponds to an individual scammer/organization who cloned their own scam contracts.} 



\section{Cluster-Aware Scam Profit}
\label{sec:cluster_aware_scam_profit}

While the presence of wash traders in Rug Pulls on Uniswap was observed as early as 2021 in Xia~\et~\cite{Xia_etal_2021} and also noted in subsequent works~\cite{Cernera_etal_USENIX2023, Sharma_etal_2023, Huynh_etal_arxiv_2023}, as far as we know, the \textit{wash-trading component} has never been accounted for in existing estimates of scam profits. One reason is that without the reconstruction of a scam cluster (or more generally, a scam network - see App~\ref{app:scam_network}), it is impossible to determine the wash-trading component. Our work seeks to address this research gap by leveraging the newly introduced scam clusters (Section~\ref{sec:clustering}) and employing the so-called \textit{cluster-aware} scam profit formula. Let us first review the commonly used profit formula for Rug Pulls and demonstrate it via a real example.

\vspace{-3pt}
\begin{definition}(Naive scam profit formula, e.g.~\cite[Sec.~7.1]{Cernera_etal_USENIX2023}) 
\label{def:existing_profit_formula}
The profit for a scam pool $p$ is $\delta(p)\triangleq Y(p)-X(p)$, where
\begin{itemize}
    \item $X(p)$ is the amount of high-value token (ETH/BNB) that the scammer addresses gave to the pool (by adding liquidity or swapping) \textit{plus} transaction fees.
    \item $Y(p)$ is the amount of high-value token (ETH/BNB) that the scammer addresses gained from the pool (by removing liquidity or swapping) \textit{minus} transaction fees.
\end{itemize} 
\end{definition}

\vspace{-5pt}
\begin{example}
\label{ex:pumpkin}
The creator of the scam token \texttt{PUMPKIN} on Uniswap (0x6C1e7FfAe984b5644C2ab95FC3aDF5794317C6aE) added $1.8$ ETH, swapped in twice totaled 0.22 ETH and removed $9.27$ ETH from the pool. Ignoring the very small transaction fees, the naive profit formula (Def.~\ref{def:existing_profit_formula}) gives $X = 1.8+0.22 = 2.02$ ETH, and $Y = 9.27$ ETH, leading to a sizable profit of $Y-X = 9.27-2.02=7.25$ ETH (more than US \$11,000) for \texttt{PUMPKIN}'s creator in less than an hour. We will see next that this estimation is far off from the real profit.
\end{example}

\vspace{-5pt}
\begin{definition}[Cluster-aware profit]
    \label{def:cluster_aware_profit_formula}
    The profit for a scam pool $p$ \textit{within a scam cluster} $C$ is $\delta(p,C) \triangleq Y(p,C)-X(p,C)$, where 
    \begin{itemize}
        \item $X(p,C)$ is the amount of high-value token (ETH/BNB) that the addresses in $C$ \textit{paid to the pool}, including liquidity additions and swap-ins, \textit{plus} transaction fees.
        \item $Y(p,C)$ is the amount of high-value token (ETH/BNB) that the addresses in $C$ \textit{gained from the pool} (by removing liquidity or swapping out), \textit{minus} transaction fees.
    \end{itemize}
    Note that $X(p,C)$ includes the term $X(p)$ in Definition~\ref{def:existing_profit_formula} and the new wash-trading component $Z(p,C)$ - the total amount of high- value token paid to the pool by the wash traders in~$C$. The \textit{total profit of a scam cluster} $C$ is the sum of profits of all scam pools in the cluster $\Delta(C)\triangleq \sum_{p \in C}\big(Y(p,C)-X(p,C)\big)-T(C)$, where $T(C)$ is the total transaction fee spent on direct high-value token (ETH/BNB) transfers among the scammer addresses in $C$.
\end{definition}

\vspace{-9pt}
\begin{example}[continued from Example~\ref{ex:pumpkin}]
\label{ex:pumpkin_continued}
After constructing the scam cluster containing \texttt{PUMPKIN}'s creator address as discussed in Section~\ref{sec:clustering}, it becomes clear that \textit{all big investors were wash traders} from that cluster (see Fig.~\ref{fig:scam_cluster_PUMPKIN}). In particular, \textbf{4b10} and \textbf{25a1} swapped in (multiple times) 3.56 and 3.54 ETH in total, respectively, while the creator \textbf{c6ae} swapped in (twice) 0.22 ETH (see Fig.~\ref{fig:scam_cluster_PUMPKIN}). The cluster-aware formula (Def.~\ref{def:cluster_aware_profit_formula}) gives $X\approx 9.1463 = 1.8177+7.3286$, $Y=9.2661$, and hence the \textit{real profit} is merely $0.1198$ ETH.
\end{example}

Note that Xia~\et~\cite{Xia_etal_2021} identified \textit{1-hop} wash traders that received direct ETH transfers from scammer addresses before their swapping of scam tokens. While this may hold for many cases, it might also lead to a potential misconception that wash-trader addresses \textit{must} receive fund from the scammer address before performing wash trading.
The cluster in Fig.~\ref{fig:scam_cluster_PUMPKIN} provides a \textit{counterexample} to such a misconception. 
Moreover, we observed that some scammer addresses like \textbf{c8ae} and \textbf{fb0b} not only wash traded for the pool created by \textbf{b116} but also transferred ETH to it, apparently received nothing back. While this behavior appears completely irrational under the traditional different-addresses-operating-independent-scams point of view, it makes perfect sense from the scam-cluster perspective: if the entire cluster was operated by a single scammer, then it doesn't matter who transferred to/wash traded for whom, as long as these operations collectively work at the cluster level. 


\textbf{Scam Profit Evaluations.} We implemented and ran the naive and cluster-aware scam profit formulas on \textit{all} Uniswap and Pancakeswap clusters obtained in Section~\ref{sec:clustering}. We found that on average, the naive scam profit \textit{inflates} the pool/cluster profit by at least 32\% on Uniswap and 24\% on Pancakeswap. More specifically, on average, a scam pool and scam cluster on Uniswap yielded a profit of 2.685 ETH and 3.96 ETH according to the naive formula, respectively. However, if the cluster-aware formula is used, then the corresponding numbers are 2.03 ETH and 2.99 ETH. On Pancakeswap, the naive formula gives 3.841 BNB and 13.102 BNB for the average pool and cluster profits, while the cluster-aware formula returns 3.109 BNB and 10.605 BNB. 
Moreover, 2,395 clusters (28\%) on Uniswap and 7,441 clusters (41\%) on Pancakeswap have wash traders. For such clusters, the averaged profits given by the naive formula are 72\% and 41\% larger than those by the cluster-aware formula.

\begin{figure}[t]
    \centering
    \includegraphics[scale=1]{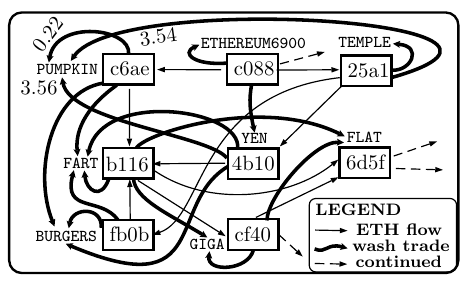}
    \vspace{-3pt}
    \caption{\textmd{Part of the scam cluster containing \textbf{c6ae}, the creator of the scam token \texttt{PUMPKIN}. The figure shows how eight scammer/wash trader addresses transferred ETH to each other (thinner arrows) and also wash traded their scam tokens (thicker arrows). The pool of \texttt{PUMPKIN} was heavily wash traded, thus inflating the perceived scam profit by more than 7 ETH.}}
    \label{fig:scam_cluster_PUMPKIN}
    \vspace{-5pt}
\end{figure}

\textbf{Extension to Scam Networks.} We observe from our datasets that sometimes scammers used non-scammer addresses in their operation. 
In such cases, the scam clusters discussed in Section~\ref{sec:clustering} fail to include them and their activities. 
Besides, the cluster-aware profit formula only captures wash trader addresses that \textit{are} scammer addresses. Therefore, it only provides an \textit{upper bound} on the true scam profit, which can only be determined if \textit{all} wash trader addresses can be identified, including non-scammer addresses. 

To address the aforementioned issues, one can extend the concept of scam cluster to \textit{scam network}. From our initial manual inspection, 
a typical scam network contains not only scammer addresses but also \textit{supporting} addresses that serve distinctive roles in the overall scam operation including \textit{wash traders}, \textit{transferrers}, \textit{depositors}, \textit{withdrawers}, and \textit{coordinators}. 
Theoretically, one can first reconstruct the scam network, treat all investor addresses in that network as wash traders, and use a \textit{network-aware} scam profit formula (a straightforward generalization of the cluster-aware formula) to obtain the \textit{true} scam profits for all pools in that network. 
However, reconstructing a scam network accurately and efficiently is a highly challenging task, 
which goes beyond the scope of this paper. Thus, we only conducted a preliminary discussion of scam networks (see Apps.~\ref{app:scam_network}) and left a more thorough investigation for future research. 

\section{Conclusions}
\label{sec:conclusion}



In this work, we explored how serial Rug-Pull scammers operated on UniswapV2 and PancakeswapV2. We detected highly coordinated \textit{scam patterns}, and examined the Rug-Pull scam operations on DEXs as part of a large-scale  \textit{cluster/network} rather than as isolated incidents. 
We note that our scam patterns were rigorously formalized in order to incorporate essential information about the Rug Pull scams, including the times the liquidity was added to and removed from the scam pools and the scam expenses and revenues. We found that such scam patterns were abundant on Uniswap and Pancakeswap, covering around 40\% of all scammer addresses on both DEXs. This finding suggests that serial scammers did follow operational patterns that can be detected.  
We also formally defined the concept of scam clusters based on native-token transfers and scam creations. Noticeably, we found that scam contracts are highly similar within the same clusters and dissimilar across clusters, supporting our view that each cluster belongs to the same scammer who cloned the contracts. Finally, we proposed a more accurate scam profit formula by including the wash traders in the cluster. By evaluating the scam profit based on the naive and the new cluster-aware formula, we found that the naive formula has inflated significantly the true scam profit as it misses the wash-trading component of each scam.

Note that we only collected and analyzed blockchain data that are publicly available. 
We have also contacted Chainabuse~\cite{chainabuse}, Etherscan~\cite{Etherscan}, BSCScan~\cite{BSCscan}, and Hashdit~\cite{Hashdit} to report our findings. 
The Python codes for our project are available online at~\cite{our_git_hub}. 

\begin{acks}
This work was supported by the Australia Research Council under the Discovery Project Grant DP200100731. 
We thank the PC members and the reviewers for very helpful comments.
\end{acks}

\bibliographystyle{splncs04}
\bibliography{Attack_of_the_Clones_WWW_FINAL}

\begin{thebibliography}{10}
\providecommand{\url}[1]{\texttt{#1}}
\providecommand{\urlprefix}{URL }
\providecommand{\doi}[1]{https://doi.org/#1}

\bibitem{our_git_hub}
Github repository (2024), \url{https://github.com/duyhuynhdev/serial-scammer-analyser}

\bibitem{bep20_standard}
Binance: {BEP-20} token standard, \url{https://academy.binance.com/en/glossary/bep-20}

\bibitem{Binance_20M_phishing}
Binance: Binance {CEO} discusses \$20 million scam attempt (2023), \url{https://www.binance.com/en-IN/square/post/910258}

\bibitem{brill2000improved}
Brill, E., Moore, R.C.: An improved error model for noisy channel spelling correction. In: Proceedings of the 38th annual meeting of the association for computational linguistics. pp. 286--293 (2000)

\bibitem{casado2018blockchain}
Casado-Vara, R., Prieto, J., De~la Prieta, F., Corchado, J.M.: {H}ow blockchain improves the supply chain: {C}ase study alimentary supply chain. Procedia computer science  \textbf{134},  393--398 (2018)

\bibitem{Cernera_etal_USENIX2023}
Cernera, F., Morgia, M.L., Mei, A., Sassi, F.: Token spammers, rug pulls, and sniper bots: {A}n analysis of the ecosystem of tokens in {Ethereum and in the Binance Smart Chain} ({{{{{BNB}}}}}). In: 32nd USENIX Security Symposium (USENIX Security 23). pp. 3349--3366 (2023)

\bibitem{chainabuse}
Chainabuse: Chainabuse: Report malicious crypto activity (2024), \url{https://www.chainabuse.com/}

\bibitem{chainalysis_report_2022}
Chainalysis: Chainalysis's crypto crime report 2022 (2022), \url{https://go.chainalysis.com/rs/503-FAP-074/images/Crypto-Crime-Report-2022.pdf}

\bibitem{chainalysis_report_2024}
Chainalysis: Chainalysis's crypto crime report 2024 (2024), \url{https://go.chainalysis.com/crypto-crime-2024.html}

\bibitem{Chen_etal_Internetware2024}
Chen, S., Chen, J., Yu, J., Luo, X., Wang, Y.: The dark side of {NFTs}: {A} large-scale empirical study of wash trading. In: Proceedings of the 15th Asia-Pacific Symposium on Internetware. p. 447–456. Internetware '24 (2024)

\bibitem{Chen_etal_NDSS2025}
Chen, Z., Hu, Y., He, B., Luo, D., Wu, L., Zhou, Y.: Dissecting payload-based transaction phishing on {Ethereum}. In: Usenix Network and Distributed System Security Symposium (NDSS) - to appear (2025)

\bibitem{coingecko_dex}
Coingecko: Decentralized exchanges (2024), \url{https://www.coingecko.com/en/exchanges/decentralized?chain=ethereum}

\bibitem{ethereum_cap}
Coinmarketcap: Ethereum martket capitalisation (2024), \url{https://coinmarketcap.com/currencies/ethereum/}

\bibitem{uniswap_forks}
Defillama: Uniswap{V}2 {F}olks (2024), \url{https://defillama.com/forks}

\bibitem{erc20_standard}
Ehereum: {ERC-20 Token Standard}, \url{https://ethereum.org/en/developers/docs/standards/tokens/erc-20/}

\bibitem{GuanLi_CCS_2024}
Guan, S., Li, K.: Characterizing {E}thereum address poisoning attack. In: Proceedings of the ACM SIGSAC Conference on Computer and Communications Security (CCS). pp. 986–--1000 (2024)

\bibitem{Hashdit}
Hashdit: Hashdit: Securing bnb chain (2024), \url{https://www.hashdit.io/}

\bibitem{Huynh_etal_arxiv_2023}
Huynh, P.D., Silva, T.D., Dau, S.H., Li, X., Gondal, I., Viterbo, E.: From programming bugs to multimillion-dollar scams: {A}n analysis of trapdoor tokens on decentralized exchanges (2023), \url{https://arxiv.org/abs/2309.04700}

\bibitem{immunefi_hack_fraud_reports}
Immunefi: Immunefi crypto losses report, \url{https://immunefi.com/research/}

\bibitem{jaiman2020consent}
Jaiman, V., Urovi, V.: {A} consent model for blockchain-based health data sharing platforms. IEEE access  \textbf{8},  143734--143745 (2020)

\bibitem{soliduslabs_washtrading_blog}
Labs, S.: {A-A} wash trading detection on {Uniswap V2}: {A} new tool for investors \& investigators (2023), \url{https://dune.com/blog/a-a-wash-trading-detection-on-uniswap-v2-a-new-tool-for-investors-investigators}

\bibitem{soliduslabs_washtrading_tool}
Labs, S.: {A-A} wash trading detector (uniswap v2) (2023), \url{https://dune.com/trenlo/a-a-wash-trading-detector}

\bibitem{levenshtein1966binary}
Levenshtein, V.I., et~al.: Binary codes capable of correcting deletions, insertions, and reversals. In: Soviet physics doklady. vol.~10, pp. 707--710. Soviet Union (1966)

\bibitem{Mazorra_etal_2022}
Mazorra, B., Adan, V., Daza, V.: Do not rug on me: {L}everaging machine learning techniques for automated scam detection. Mathematics  \textbf{10}(6) (2022)

\bibitem{Morgia_etal_ICDCS_2023}
Morgia, M.L., Mei, A., Mongardini, A.M., Nemmi, E.N.: A game of {NFTs}: {Characterizing NFT} wash trading in the {E}thereum blockchain. In: 2023 IEEE 43rd International Conference on Distributed Computing Systems (ICDCS). pp. 13--24 (2023)

\bibitem{Nguyen_etal_2023}
Nguyen, M.H., Huynh, P.D., Dau, S.H., Li, X.: Rug-pull malicious token detection on blockchain using supervised learning with feature engineering. In: Proceedings of the 2023 Australasian Computer Science Week. pp. 72–--81. ACSW '23 (2023)

\bibitem{novo2018blockchain}
Novo, O.: {B}lockchain meets {I}o{T}: {A}n architecture for scalable access management in {I}o{T}. IEEE internet of things journal  \textbf{5}(2),  1184--1195 (2018)

\bibitem{openzeppelin}
OpenZeppelin: {Smart Contracts Library} (2024), \url{https://www.openzeppelin.com/solidity-contracts}

\bibitem{pancakeswap}
Pancakeswap: \url{https://pancakeswap.finance/}

\bibitem{paterson1994longest}
Paterson, M., Dan{\v{c}}{\'\i}k, V.: Longest common subsequences. In: International symposium on mathematical foundations of computer science. pp. 127--142. Springer (1994)

\bibitem{LePennecFiedlerAnte_FRL_2021}
Pennec, G.L., Fiedler, I., Ante, L.: Wash trading at cryptocurrency exchanges. Finance Research Letters  \textbf{43},  101982 (2021)

\bibitem{Sharma_etal_2023}
Sharma, T., Agarwal, R., Shukla, S.K.: Understanding rug pulls: {A}n in-depth behavioral analysis of fraudulent {NFT} creators. ACM Trans. Web  \textbf{18}(1) (oct 2023)

\bibitem{Szabo_1997}
Szabo, N.: Formalizing and securing relationships on public networks. First Monday  \textbf{2}(9) (1997)

\bibitem{bsc_verified_contract}
{E}therscanners team, T.: {BNB} {V}erified {C}ontract (2015), \url{https://bscscan.com/contractsVerified}

\bibitem{BSCscan}
{E}therscan team, T.: {BSC}scan (2015), \url{https://bscscan.com}

\bibitem{BSCscanAPI}
{E}therscan team, T.: {BSC}scan {API} (2015), \url{https://bscscan.com/apis}

\bibitem{Etherscan}
{E}therscan team, T.: Etherscan (2015), \url{https://etherscan.io}

\bibitem{EtherscanAPI}
{E}therscan team, T.: Etherscan {API} (2015), \url{https://etherscan.io/apis}

\bibitem{eth_verified_contract}
{E}therscanners team, T.: Etherscan {V}erified {C}ontract (2015), \url{https://etherscan.io/contractsVerified}

\bibitem{uniswap}
Uniswap: \url{https://uniswap.org/}

\bibitem{VictorWeintraud_WWW21}
Victor, F., Weintraud, A.M.: Detecting and quantifying wash trading on decentralized cryptocurrency exchanges. In: Proceedings of the Web Conference 2021. pp. 23–--32. WWW '21 (2021)

\bibitem{VonWachter_etal_FCWorkshop22}
von Wachter, V., Jensen, J.R., Regner, F., Ross, O.: {NFT} wash trading: {Q}uantifying suspicious behaviour in {NFT} markets. In: Financial Cryptography and Data Security. FC 2022 International Workshops: CoDecFin, DeFi, Voting, WTSC, Grenada, May 6, 2022, Revised Selected Papers. p. 299–311 (2023)

\bibitem{web3}
Web3.py: Python library for interacting with ethereum (2024), \url{https://web3py.readthedocs.io/en/stable/}

\bibitem{werner2022sok}
Werner, S., Perez, D., Gudgeon, L., Klages-Mundt, A., Harz, D., Knottenbelt, W.: Sok: {D}ecentralized finance (defi). In: Proceedings of the 4th {ACM} Conference on Advances in Financial Technologies. pp. 30--46 (2022)

\bibitem{jaccard}
Wikipedia: Jaccard index (2024), \url{https://en.wikipedia.org/wiki/Jaccard_index}

\bibitem{Xia_etal_2021}
Xia, P., Wang, H., Gao, B., Su, W., Yu, Z., Luo, X., Zhang, C., Xiao, X., Xu, G.: Trade or trick? {D}etecting and characterizing scam tokens on {Uniswap} decentralized exchange. Proc. ACM Meas. Anal. Comput. Syst.  \textbf{5}(3) (2021)

\bibitem{Yan_etal_Cybersecurity_2023}
Yan, C., Zhang, C., Shen, M., Li, N., Liu, J., Qi, Y., Lu, Z., Liu, Y.: Aparecium: {U}nderstanding and detecting scam behaviors on {E}thereum via biased random walk. Cybersecurity  \textbf{6} (10 2023)

\bibitem{Zhou_etal_ICSE_SEIP_2024}
Zhou, Y., Sun, J., Ma, F., Chen, Y., Yan, Z., Jiang, Y.: Stop pulling my rug: {E}xposing rug pull risks in crypto token to investors. In: Proceedings of the 46th International Conference on Software Engineering: Software Engineering in Practice. pp. 228–--239. ICSE-SEIP '24 (2024)

\end{thebibliography}

\appendix 

\section{Related Works}
\label{sec:related_works}

\subsection{Wash Trading}

\textit{Wash trading} activities in the cryptocurrency ecosystem  have been observed and studied in different contexts. For example, wash trading can be carried out to artificially increase the trading volume of cryptocurrency exchanges~\cite{LePennecFiedlerAnte_FRL_2021, VictorWeintraud_WWW21} to influence the perception of their popularity. Wash trading can also be used to boost the trading volume of a Non-Fungible Token (NFT) to reap the reward from an NFT marketplace or inflate its price for reselling~\cite{Morgia_etal_ICDCS_2023,VonWachter_etal_FCWorkshop22,Chen_etal_Internetware2024}. Wash trading activities were also observed in the context of Rug Pulls of ERC-20 tokens on Uniswap in Xia~\et~\cite[Sect.~5.3.3]{Xia_etal_2021}, under their investigation of collusion addresses. 
However, under their heuristics, only wash traders that have a \textit{direct} transfer with the scammers, i.e. 1-hop neighbors, are included. 
Even the recently developed \textit{A-A Wash-Trading Detector} for UniswapV2 on Dune from SolidusLab~\cite{soliduslabs_washtrading_tool,soliduslabs_washtrading_blog} can only detect self wash trading (or 0-hop wash trader). 
This simple wash-trading model fails to capture more sophisticated scams in which wash traders are multiple hops away, as observed in our datasets (see, e.g. the example in Fig.~\ref{fig:wt_9aa6}, App.~\ref{app:multi_hop_non_scammer_WT}). 

We note that most research on wash trading for NFTs in the literature (see, e.g.~\cite{VonWachter_etal_FCWorkshop22,Morgia_etal_ICDCS_2023,Chen_etal_Internetware2024} and the references therein) studied a different setting in which accounts trade NFTs \textit{directly} among themselves or via a \textit{centralized} marketplace like OpenSea or LooksRare, not via an exchange pool.  
For example, in Morgia~\et~\cite{Morgia_etal_ICDCS_2023}, a graph is built for each NFT with nodes being the addresses \textit{directly} interacting with the NFT and edges being their \textit{direct} trade of the NFT. Then, suspicious wash-trading groups include addresses that perform self-trade or strongly connected components that have a common funder or beneficiary, or maintain a zero-risk position (zero balance after all the transactions, factoring out the gas fees). The concept of zero-risk position is \textit{irrelevant} to the exchange pool setting of fungible ERC-20/BEP-20, in which investors trade \textit{with the pool} and \textit{not} among themselves. Also, the edges of the connected components in the wash-trading graph in~\cite{Morgia_etal_ICDCS_2023} represent NFT tradings, which doesn't exist in our setting as Uniswap/Pancakeswap investors do not trade ERC-20/BEP-20 tokens after buying them from the pool. 

\subsection{Scam Patterns and Clusters}

Xia~\et~\cite{Xia_etal_2021} investigated a related concept of \textit{collusion} addresses on a scam pool, which consist of a) the scammer addresses that created the scam token/pool, b) the liquidity-provider addresses that added liquidity to the pool after being funded by the scammer addresses, c) the liquidity-remover addresses that removed the liquidity from the pool and then transferred ETH/valuable tokens to the scammer addresses, d) the investor addresses that were funded by the scammer addresses before buying the scam token or transferred ETH/valuable tokens to the scammer addresses after gaining a profit trading with the pool. 

While collusion addresses \cite{Xia_etal_2021} form a part of the scam cluster/network studied in our work, they are limited to \textit{individual} scam pools and do not capture the more general serial-scammer setting where \textit{multiple} addresses operate on \textit{multiple} related scam pools. Xia~\et{} also noted that there might be more complex networks (e.g. to launder their fund) operating behind the scams that their analysis failed to capture, and hence there might be many more scam addresses not identified by their heuristics (see~\cite[Sect.~6.3]{Xia_etal_2021}). \rev{Our work partly addresses their concern by studying scam patterns and scam clusters/networks. More specifically, the investigation of scam patterns identifies some typical ways scammer addresses got funded, e.g. by a common coordinator (in an OUT or IN/OUT star) or by other scammer addresses (in a simple chain or major flow), as well as how they handled their scam profits, e.g. sending the scam money to an aggregation address (in an IN or IN/OUT star) or funding other scammer addresses (in a simple chain or major flow). Moreover, the knowledge of scam clusters/networks show how scammers use multiple addresses to operate multiple related scams and perform wash trading among themselves.} 

Sharma \et~\cite{Sharma_etal_2023} studied in depth 760 NFT Rug Pulls from various NFT marketplaces and also examined what they called repeated Rug Pulls organized by the same scammer groups, and also network of scammers, which overlap with our concept of scam clusters/networks. The star and chain structures among the scammer addresses were also observed and briefly discussed, though not formalized. The reasonably small number of scams allowed some of their algorithms to involve manual inspection (e.g. Algorithm~1), which is impossible for our large datasets.
Wash trading was left for future study in their work.

\section{Functions and Events of ERC/BEP-20}
\label{app:functions_events}

Table~\ref{tab:token_standard} shows the set of common functions and events that ERC-20 tokens and BEP-20 tokens must follow.

\begin{table}[htb]
\centering
\setlength{\tabcolsep}{4pt}
\begin{tabular}{|c|c|p{4.3cm}|}
\hline
\textbf{Type}     &   \textbf{Signature}       &     \textbf{Description}    \\ \hline
\multirow{9}{*}{Method} & name()  & Getting name of the token (e.g., Dogecoin)\\
 & symbol() &  Get symbol of the token (e.g., DOGE)\\
 & decimals() & Get the number of decimals the token uses\\
 & totalSupply() & Get the total amount of the token in circulation\\
 & balanceOf() & Get the amount of token owned by given address\\
 & \textbf{transfer()} & Transfer amount of tokens to given address from message caller\\
 & \textbf{transferFrom()} & Transfer amount of tokens between two given accounts\\
 & approve() & Allow a \textit{spender} spend token on behalf of \textit{owner} \\
 & allowance() & Get amount that the \textit{spender} will be allowed to spend\\
\hline
\multirow{2}{*}{Event} & Transfer()  & Trigger when tokens are transferred, including zero value transfers.\\
 & Approval() & Trigger on any successful call to \textit{approve()}\\
 \hline
\end{tabular}
\vspace{5pt}
\caption{\textmd{Functions and events of the ERC-20 (Ethereum) and BEP-20 (BNB Smart Chain) standard. Among required functions, \texttt{transfer()} and \texttt{transferFrom()} are two basic functions for digital assets transferring.}}
\label{tab:token_standard}
\end{table}

\section{Scam Datasets Analysis}
\label{app:data_collection}

We observed that only 4.6\% scammer addresses on Uniswap created more than one scam pools. The largest number of scam pools (162) was created by 0xee80ba78889acf579b6a8d503470b4a67cf77076. 
However, on Pancakeswap, nearly 20\% of scammer addresses 
created multiple scam pools, accounting for more than 288,000 pools, corresponding to 61.2\% of all scam pools.
Especially, the scammer address 0x608756c184a0723077b0c10f97f4d054c9ee1c0f created more than 18,000 different scam pools in more than 3 years (the latest scam was created in November 2024). 
This address generated several meme tokens following 
hot events to attract the attention of users. For example, during the US presidential election, this account created 16 \texttt{TRUMP} coins and listed on Pancakeswap. 
Noticeably, it used almost identical source code to create many tokens, which is a cost-effective way for scam mass production. 
We further investigated 4,002 scammer addresses on UniswapV2 and 28,345 scammers on PancakeswapV2 that created more than one scam pools. Using our \textit{AST-Jaccard score}, which bypasses standard ofuscation techniques such as adding comments and spaces, changing variable names, or reordering code (see Appendix~\ref{app:similarity_score}), to measure similarity among given contracts, we found that 77.6\% and 74.5\% of such addresses on Uniswap and Pancakeswap, respectively, deployed contracts with over 70\% similarity. Moreover, 1,971 Uniswap scammer addresses and 10,236 Pancakeswap scammer addresses reused almost the same contracts (100\% AST-Jaccard similarity) for all the deployed scam tokens. 
Especially, 0x2c1eb6ca34997f6601cffe791831ad6d5cb9e937 on Uniswap and 0x50cd702d4b11cf6acec84bba739a3fd3a460e4e4 on Pancakeswap deployed almost identical contract 150 times and 1,191 times, respectively.


\section{Contract Similarity}
\label{app:similarity_score}
We discuss in this appendix the preprocessing step and the construction of a new contract similarity score that can bypass a number of obfuscation techniques when comparing token contracts.

Popular string metrics used for measuring the source code similarity between two contracts include
Levenshtein (edit) distance~\cite{levenshtein1966binary}, longest common subsequences~\cite{paterson1994longest}, and Damerau-Levenshtein distance~\cite{brill2000improved}. However, such metrics are not only expensive to compute over a large number of contract pairs but also susceptible to simple code obfuscation techniques. For example, to reduce the similarity score between two contract clones, scammers can employ techniques like (M1) comments and spaces insertion/deletion, (M2) identifier/variable names modification, or (M3) code reordering to increase the edit distance arbitrarily. To tackle these, we propose a token-based similarity score call AST-Jaccard score, which leverages source code's AST, a collision-resistant hash function, and Jaccard similarity to compare scam contracts. 

More specifically, as AST keeps the syntax and sematic information of the source code while ignoring non-essential information such as spaces, comments, specific function and variable names, using AST instead of the source code itself tackles (M1) and (M2). To bypass (M3) code reordering, e.g. swapping functions or variables around, the new score uses the Keccak-256 hashing algorithm to hash all tokens/components in the AST of each source code and group them into a set. As the order of elements is not important for a set, even when the variables and function were rearranged, the Jaccard similarity on sets still picks up the right overlapping ratio between the two sets corresponding to two contract codes. 
To further improve the accuracy, we also applied a preprocessing steps to remove all common libraries and interfaces in every source code, which account for 40\%-50\% of the code itself and would interfere with the score. More details about the preprocessing and the AST-Jaccard score can be found in Appendix~\ref{app:similarity_score}.

\textbf{Common Libraries and Interfaces Removal.} Fungible tokens are implemented by following the common standard (e.g ERC-20 or BEP-20 interface) so these tokens often contains the common codes from the interface, inflating their similarity. The similarity continues to increase if they used the same common libraries such as \texttt{SafeMath},  \texttt{Address}, or \texttt{Ownable}. Hence, to avoid this inflation, we remove all common libraries and interfaces from tokens before measuring their similarities. To that end, we first collect all common libraries on OpenZeppelin~\cite{openzeppelin}, an open-source framework for writing secure and scalable smart contracts. Then we remove all these libraries and interfaces from token's contracts before doing tokenisations.

\textbf{Code Tokenisation.} Our approach first parses the source code of a contract to an AST, which keeps the syntax and semantic information of a contract only. Subsequently, we extract a token from a type of each node in the AST (e.g., \texttt{Mapping}, \texttt{IfStatement}, \texttt{BinaryOperation}, \texttt{Assignment}). In this manner, some inessential information such as spaces, comments, function names, variable names and their values will be eliminated. Thus, we can solve M1 and M2. We run an extraction for each component in a contract, i.e. \textit{state variables}, \textit{functions}, \textit{events}, and \textit{modifiers}. Notably, code lines in a component will be parsed to an array of tokens. For example, we can get correspondingly an array of \textit{IfStatement}, \textit{BlockIdentifier}, \textit{IndexAccess}, \textit{BinaryOperation}, and \textit{Literal} tokens from the \texttt{if} condition ``\texttt{if(sender[i] == `0x0d83a1')}''. In other words, each array of tokens marks the semantic information of each component in a contract.

\textbf{Integration of AST, hash function, and Jaccard similarity Score.} Next, we try to overcome M3 at a contract level, i.e. to
identify a clone of a contract even if the scammer has reordered different components in the contract. 
To improve our algorithm's performance, we concatenate all tokens in the token array of a component as a string and apply a Keccak-256 hash algorithm. As a result, each component in the contract will be represented as a unique hash (256 bits). Next, we employ the classic Jaccard index~\cite{jaccard} for two corresponding sets of hashes as in \eqref{jaccard}, 
where $J(A, B)$ is the similarity between two contracts, $A$ is a list of hashes of the first contract and $B$ is a list of hashes of the second contract. 
  \begin{equation}
    \label{jaccard}
         J(A,B)=\frac{|A \cap B|}{|A \cup B|}=\frac{|A \cap B|}{|A| + |B| - |A \cap B|}.
    \end{equation}

The overall algorithm is presented below.

\floatname{algorithm}{Procedure}
\begin{algorithm}[htb]
\caption{\textbf{tokenization}($\textsf{token}$)}
\begin{algorithmic}[1]
        \STATE $\textsf{ast} \gets $ \textbf{parseAST}($\textsf{token.contract}$)
        \STATE  $\textsf{components} \gets $\textbf{removeCommon}($\textsf{ast.contract.nodes}$)
        \STATE  $\textsf{hashes} \gets [] $
        \FOR{$ \textsf{j} \gets 0$ to $\textsf{components.length}$}
            \STATE $\textsf{syntactic\_tonkens}\gets $ \textbf{extractTokens}($\textsf{components[i]}$)
            \STATE $\textsf{token\_str} \gets $ \textbf{concat}($\textsf{syn\_tonkens}$)
            \STATE $\textsf{hashes[i]}\gets$ \textbf{keccak256}($\textsf{token\_str}$)
        \ENDFOR
        \RETURN $\textsf{hashes}$
\end{algorithmic}
\end{algorithm}
\vspace{-5pt}

\floatname{algorithm}{Procedure}
\begin{algorithm}[htb]
\caption{\textbf{similarity}$(\textsf{tokenA}, \textsf{tokenB})$}
\begin{algorithmic}[1]
    \STATE $\textsf{hashesA} \gets $\textbf{tokenization}($\textsf{tokenA}$)
    \STATE $\textsf{hashesB} \gets $\textbf{tokenization}($\textsf{tokenB}$)
    \STATE $\textsf{intersection} \gets \textsf{hashesA}$ $\&$ $\textsf{hashesB}$
    \STATE $\textsf{union} \gets  \textsf{hashesA}$ $|$ $\textsf{hashesB}$
    \STATE $\textsf{jaccard\_index} \gets \textsf{intersection.length}$ / $\textsf{union.length}$
    \RETURN $\textsf{jaccard\_index}$
\end{algorithmic}
\end{algorithm}
\vspace{-5pt}





\section{Scam Chains: More Analysis}
\label{app:scam_chain}

The longest simple scam chain on Uniswap starts at the address 0x79daa9236e6825f023ab3ebd2cecfe94b48789d1 and ends at the address 0x3ba3a256cc8b13cdba3ad6a34352e0cfb51bbbe0, over a period of 26 days with 274 scammer addresses. Note that \textbf{89d1} has another scammer address \textbf{7ec2} as its largest funder. However, the chain doesn't include \textbf{7ec2} since \textbf{89d1} is not its largest beneficiary. 

The longest simple scam chain on Pancakewap starts at the address 0x39b81c24b8aaf10182d63706e940b8994859866d and ends at the address 0x16fdd937c56c4f82843487a0d28c34b8ae100d49, over a period of 38 days with 713 scammer addresses. 
Although the first address was funded by two other scammer addresses, the current definition of the chain doesn't capture any of these addresses. The reason is that it was funded twice to run two separate scams, and the first transfer had a smaller amount, hence is not a maximum in-transaction  (violating (C1)), whereas the larger transfer happened after the first scam, and hence was ignored (violating (C2)). Both the definitions of simple chain and major flow assume that scam funding transfers all happened before the first scam. Thus, they miss the case when the scammer address performed multiple scams and was funded right before the start of each scam. Note that all such cases will still be captured by a scam cluster because a cluster takes into account the transfers without considering their timings.   

We also measured the AST-Jaccard similarity scores for token contracts in scam chains that have at least two contracts available on Etherscan and BSCscan. On average, Uniswap chains have intra-chain similarity score of 83\% and inter-chain similarity score of 23\%. Moreover, 1,569 (49\%) of such Uniswap chains have 100\% intra-chain similarity score, and 3,367 (97\%) have inter-chain similarity scores at most 30\%. On average, Pancakeswap chains have intra-chain score of 80\% and inter-chain 21\%. Moreover, 3,687 (33\%) of such chains have intra-chain similarity score of 100\%, and 11,888 (100\%) has inter-chain similarity scores at most 30\%.  

\vspace{-10pt}
\section{Scam Stars: Detection \& Analysis}
\label{app:star}

\textbf{Description of} \texttt{StarDetector}. 
First, for each scammer address $s$, the algorithm identifies its \textit{major funder} as the \textit{in-neighbour} $f$ that satisfies the following conditions: F1) $f$ funded at least 100\% of the cost of the first scam carried out by $s$ with a single in-transaction before the first scam, and F2) $f$ must be the (strictly) largest funder of $s$. Similarly, \texttt{StarDetector} identifies the \textit{major beneficiary} of $s$ as the \textit{out-neighbour} $b$ that satisfies the following conditions: B1) $s$ transferred at least $p=90\%$ of its last scam's revenue to $b$ in a single out-transaction after the last scam, and B2) the corresponding out-transaction must be the largest out-transaction after the last scam was conducted. 

Next, \texttt{StarDetector} identifies the star type(s) that $s$ potentially belongs to as follows: if $f \equiv b$ then $s$ may belong to a single IN/OUT-star with center $f$; otherwise, if $f$ never received any fund from $s$ then $s$ may belong to an OUT-star with center $f$, and if $b$ never funded $s$ then $s$ may belong to an IN-star with center $b$. Note that it is possible that $s$ belongs to both an IN-star and an OUT-star (with different centers). For each star type and the corresponding potential center $c=f$ or $c=b$ or $c=f\equiv b$, the algorithm then examines each of $c$'s neighbors and checks whether $c$ is also their potential star center of the corresponding type. If there are $n\geq 5$ neighbors of $c$ (including the original scammer address $s$), then \texttt{StarDetector} returns the corresponding star, and then repeats with other addresses. The algorithm also keeps track of the star type each address already belongs to for avoiding redundant work.

The largest star on Uniswap is an IN-star with 585 has scammer addresses, all transferred their scam revenues to the center/depositor address  0x4a27bae40a12c0cfe8ed23d18314d96676b45f3a, which frequently deposited to the CEX Coinbase. This star operated over a period of 181 days, with a total fund-in of 1,314 ETH. Each scammer address was funded by Coinbase, performed exactly one scam and transferred fund to the center. The longest periods for the OUT and IN/OUT stars on Uniswap are almost the same (identical after rounding) because the corresponding stars, which have the same center, performed their last set of scams on the same day in June 2023.

The largest star on Pancakeswap is an IN/OUT-star with 9,012 scammer addresses, all received fundings and transferred scam revenues to the center 0xd959faa53914d150e695caa32ee5afe5048cbe70.
This star operated over a period of 178 days, with 592,168 BNB coming out of the center and 939,528 BNB coming in to the center. One satellite address is 0x7e546054822a57c6cd0beb4050353a1e4c0122e4, which ran only one scam on 3/10/2023. The total number of scam pools created was 9,183, slightly larger than the number of scammer addresses, which means that some addresses created more than one scam pools.

We repeated the detection of stars with a few different thresholds $p = 80$-$95\%$, and found that the statistics (see Tables~\ref{tab:more_stats_star_Uniswap} and~\ref{tab:more_stats_star_Pancakeswap}) are mostly similar to when $p=90\%$. 

\begin{table}[htb!]
\centering
\setlength{\tabcolsep}{1pt}
\begin{tabular}{|l|r|r|l|l|l|l|l|}
\hline
\textbf{Type} & \multicolumn{1}{l|}{\textbf{$p$}} & \multicolumn{1}{l|}{\textbf{\#Star}} & \textbf{Size} & \textbf{FundIn} & \textbf{FundOut} & \textbf{Period} & \textbf{\#Scams} \\ \hline
IN            & 80\%                             & 1622                                 & 574;19        & 17597;108        &                   & 925;97          & 574;20           \\ \hline
I/O       & 80\%                             & 71                                   & 150;15        & 22447;572        & 22734;555         & 476;48          & 150;16           \\ \hline
OUT           & 80\%                             & 58                                   & 62;9          &                  & 19646;485         & 476;55          & 62;11            \\ \hline
IN            & 85\%                             & 1603                                 & 570;19        & 17597;109        &                   & 925;96          & 570;20           \\ \hline
I/O       & 85\%                             & 72                                   & 157;15        & 22447;564        & 22734;547         & 476;54          & 157;16           \\ \hline
OUT           & 85\%                             & 60                                   & 58;9          &                  & 19646;472         & 476;54          & 58;11            \\ \hline
IN            & 90\%                             & 1575                                 & 585;19        & 17597;104        &                   & 925;96          & 585;20           \\ \hline
I/O       & 90\%                             & 73                                   & 159;15        & 22447;557        & 22734;540         & 476;56          & 159;16           \\ \hline
OUT           & 90\%                             & 61                                   & 66;10         &                  & 19646;465         & 476;56          & 66;11            \\ \hline
IN            & 95\%                             & 1574                                 & 556;19        & 17597;109        &                   & 925;96          & 556;20           \\ \hline
I/O       & 95\%                             & 72                                   & 153;15        & 22447;560        & 22734;543         & 476;54          & 153;16           \\ \hline
OUT           & 95\%                             & 60                                   & 63;9          &                  & 19646;470         & 476;56          & 63;11            \\ \hline
\end{tabular}
\caption{\textmd{Statistics for \textbf{Uniswap} scam stars when the threshold $p$ varies between $80\%$ and $95\%$.}}
\label{tab:more_stats_star_Uniswap}
\end{table}

\begin{table}[htb]
\centering
\setlength{\tabcolsep}{0.5pt}
\begin{tabular}{|l|l|l|l|l|l|l|l|}
\hline
\textbf{Type} & \multicolumn{1}{l|}{$p$} & \multicolumn{1}{l|}{\textbf{\#Star}} & \textbf{Size} & \textbf{F. In} & \textbf{F. Out} & \textbf{Period} & \textbf{\#Scams} \\ \hline
IN            & 80\%                             & 1310                                 & 977;21        & 5734;79          &                   & 1150;144        & 1008;27          \\ \hline
I/O       & 80\%                             & 266                                  & 9012;56       & 939528;4052      & 592168;2548       & 1025;51         & 9183;63          \\ \hline
OUT           & 80\%                             & 351                                  & 1900;28       &                  & 6513;98           & 1184;175        & 2824;43          \\ \hline
IN            & 85\%                             & 1309                                 & 977;21        & 5734;79          &                   & 1150;146        & 1008;27          \\ \hline
I/O       & 85\%                             & 256                                  & 9012;58       & 939528;4209      & 592168;2646       & 1025;52         & 9183;65          \\ \hline
OUT           & 85\%                             & 341                                  & 1900;28       &                  & 6513;101          & 1184;178        & 2824;44          \\ \hline
IN            & 90\%                             & 1301                                 & 977;22        & 5734;79          &                   & 1150;146        & 1008;27          \\ \hline
I/O       & 90\%                             & 251                                  & 9012;59       & 939528;4293      & 592168;2699       & 1025;51         & 9183;66          \\ \hline
OUT           & 90\%                             & 339                                  & 1900;28       &                  & 6513;101          & 1184;180        & 2824;44          \\ \hline
IN            & 95\%                             & 1293                                 & 977;22        & 5734;79          &                   & 1150;146        & 1008;27          \\ \hline
I/O       & 95\%                             & 251                                  & 9012;59       & 939528;4293      & 592168;2699       & 1025;51         & 9183;66          \\ \hline
OUT           & 95\%                             & 338                                  & 1900;28       &                  & 6513;101          & 1184;181        & 2824;44          \\ \hline
\end{tabular}
\caption{\textmd{Statistics for \textbf{Pancakeswap} scam stars when the threshold $p$ varies between $80\%$ and $95\%$.}}
\label{tab:more_stats_star_Pancakeswap}
\end{table}

We also measured the AST-Jaccard similarity scores for token contracts in scam stars (with $p=90\%$) that have at least two contracts available on Etherscan and BSCscan. On average, Uniswap stars have intra-star similarity score of 60\% and inter-star similarity score of 27\%. Moreover, 505 (35\%) of such Uniswap stars have at least 70\% intra-star similarity score, and 910 (60\%) have inter-star similarity scores at most 30\%. On average, Pancakeswap stars have intra-star score of 65\% and inter-star 24\%. Moreover, 976 (44\%) of such stars have intra-star similarity score of 100\%, and 1,858 (82\%) have inter-star similarity scores at most 30\%.  

\section{Overlaps of Scam Patterns}
\label{app:scam_patterns_more_analysis}

\rev{Setting $p=90\%$, we found that 693 addresses on Uniswap (0.48\%) and 2,917 
addresses on Pancakeswap (1.2\%) 
belong to both a chain and a star.
In the Uniswap scammer dataset, 16,367 (11.2\%) addresses belong to both a chain and a major flow, and 785 (0.5\%) addresses belong to both a star and a major flow. For the Pancake dataset, 49,170 (20.6\%) addresses belong to both a chain and a major flow, and 3,469 (1.5\%) addresses belong to both a star and a major flow. It is clear that most addresses that belong to major flows also belong to simple chains. This is because most of the maximal major flows have width two and are simple chains (but might not be maximal chains). There are 210 major flows (4\% of all flows) on Uniswap and 354 major flows (2\%) on Pancakeswap that are not simple chains. This suggests that simpler funding patterns were more widely used perhaps due to easier address management. The existence of more complicated major flows, although a small fraction, indicates that some scammers were willing to go the extra miles to make their scam operations less obvious and potentially harder to trace. We expect that scam patterns will evolve and get more and more complex over time to elude tracking by scam detection tools.}

\section{Major Flow: Detection \& Analysis}
\label{app:majority_flow}

We present below the main steps of \mfd, which identifies all major flows within a given set of scammer addresses.

\begin{algorithm}
\begin{flushleft}
\textbf{Algorithm} \mfd$(S)$
\vspace{1pt}
\hrule
\vspace{1pt}
\noindent \textbf{Step 1.} Construct $T_F(s)$, $T_B(s)$, $F(s)$, and $B(s)$ following Def.~\ref{def:major_transactions} for all $s \in S$. Remove all $s$ from $S$ where both $T_F(s)$ and $T_B(s)$ do not exist, or $F(s)\not \subseteq S$ and $B(s)\not \subseteq S$.   

\noindent \textbf{Step 2.} Construct all minimal major flows $\{G_i = (V_i,T_i)\}_{i=1}^m$ as follows. First, let $\T$ be the set collection
\begin{multline*}
\T \triangleq \{T_F(s)\colon s \in S, \varnothing \neq T_F(s) \subseteq S\}\\ \cup \{T_B(s)\colon s \in S, \varnothing \neq T_B(s) \subseteq S\}.
\end{multline*}
Initially, all sets in $\T$ are \textit{unused}.
For each iteration $i=1,2,\ldots$, set $T_i \triangleq T_F(s)$ or $T_B(s)$ where $T_F(s)$ or $T_B(s)$ is an unused set in $\T$. Keep expanding $T_i$ by merging it with other unused sets in $\T$ using the following merging rule: if $T_i$ contains $T_F(s')$ then merge it with $T_B(f)$ for all $f \in F(s')$, whereas if it contains $T_B(s')$ then merge with $T_F(b)$ for all $b \in B(s')$, and label those sets $T_B(f)$ or $T_F(b)$ as \textit{used}. Stop when no further expansion is possible. Let $V_i \triangleq \cup_{t=(u,v) \in T_i}\{u,v\}$, which consists of all the scammer addresses involved in the transactions in $T_i$, return $G_i=(V_i,T_i)$, and repeat for the next iteration $i+1$ until all sets in $\T$ have been used. 

\noindent \textbf{Step 3.} Build an auxiliary graph $\G=(\V,\E)$ with the vertex set $\V \triangleq \{G_1,G_2,\ldots,G_m\}$ and edge set $\E \triangleq \{(G_i,G_j)\colon i \neq j, V_i\cap V_j \neq \varnothing\}$. Find all the connected components $C_1,\ldots,C_k$ of $\G$ and return $\{H_j\triangleq \cup_{G_i \in V(C_j)} G_i\colon j = 1,\ldots,k\}$ as maximal major flows in $S$.
\end{flushleft}
\end{algorithm}

We now provide a proof for Theorem~\ref{thm:majority_flow}, which guarantees the correctness and efficiency of \mfd. 

\begin{proof}[Proof of Theorem~\ref{thm:majority_flow}]
\textbf{Proof of part (a)}. 
Let $\Gu = (\Vu = V\cup V',\Tu = T \cup T')$ be the union of $G$ and $G'$, and assume that there exists $v \in V\cap V' \neq \varnothing$. Clearly, $\Gu$ satisfies (P2) because there is a path between any two distinct vertices $u, w$ in $\Gu$, which is the union of the path from $u$ to $v$ and then the path from $v$ to $w$ in the individual clusters (if $u\equiv v$ or $w\equiv v$ then one of such path is empty). Moreover, for any $s \in \Vu$, as both $G$ and $G'$ satisfy (P1), we have $I_\Gu(s) = I_G(s)\cup I_{G'}(s) = T_F(s)$ if either $I_G(s)$ or $I_{G'}(s)$ is nonempty. A similar conclusion holds for $O_\Gu(s)$. Thus, $\Gu$ satisfies (P1) as well.

\textbf{Proof of Part (b)}. We aim to show that each $G_i$, $i=1,\ldots,m$, is a sub-graph of a minimal major flow, and is a major flow itself, hence coincides with a minimal cluster. As the algorithm goes through every funder-beneficiary transactions within $S$ (captured by the sets in $\T$), Step~2 of the algorithm must return all the minimal major flows in $S$. 

First, suppose $T_i = T_F(s)$ or $T_i = T_B(s)$ initially for some $s \in S$ in Iteration $i$, $i = 1,\ldots,m$. Let $G(s)=(V(s),T(s))$ be the \textit{minimal} major flow containing $T_F(s)$ or $T_B(s)$, respectively. Clearly, $G_i = (V_i,T_i)$ is a subgraph of $G(s)$. We prove in an inductive manner that $G_i$ is always a subgraph of $G(s)$ after each merging step. Indeed, assume that $G(s)$ currently contains $G_i$ as a sub-graph. If $T_i$ contains $T_F(s')$ for some $s'\in S$ then it is merged with $T_B(f)$ for all $f \in F(s')$. For such $f$, since $(f,s') \in T_F(s') \subseteq T(s)$, we have $(f,s') \in O_{G(s)}(f)$, which implies that $O_{G(s)}(f)\neq \varnothing$, and hence, $O_{G(s)}(f)=T_B(f)$ due to (P1). Thus, $G(s)$ contains $G_i$ as a subgraph after $T_i$ is merging with $T_B(f)$ for every $f\in T_F(s')$. The same conclusion holds if $T_i$ contains $T_B(s')$ and gets merged with $T_F(b)$ for all $b \in B(s')$. Thus, $G_i$ is a subgraph of a minimal major flow.

It remains to show that $G_i$'s generated at Step~2 of the algorithm are major flows, $i=1,\ldots,m$. We first prove that $G_i$ satisfies (P1). Assume that $I_{G_i}(s) \neq \varnothing$ for some $s \in V_i$, we need to show that $I_{G_i}(s) \equiv T_F(s)$ (the case $O_{G_i}(s) \neq \varnothing$ can be treated similarly). Considering one edge $(f,s) \in I_{G_i}(s)$. As $G_i$ is constructed by adding either $T_F(\cdot)$ or $T_B(\cdot)$ one at a time, the existence of $(f,s)$ means that $G_i$ either already contains $T_F(s)$ or $T_B(f)$. In the latter, according to the merging rule, as $s \in B(f)$, $G_i$ will be further expanded by adding the edges from $T_F(s)$. Therefore, in either cases, $G_i$ would contain $T_F(s)$ eventually, hence satisfying (P1). 

\textbf{Proof of Part (c)}. We now show that Step~3 in the algorithm \mfd{} returns all maximal major flows in $S$.

First, due to Part (a), $H_j \triangleq \cup_{G_i \in V(C_j)} G_i$ is a major flow for all $j=1,\ldots,k$ as it is the union of major flows. Note that $H_i$ can be formed by uniting two major flows that have at least one vertex in common a number of times. It remains to prove that $H_j$ is maximal. For the sake of contradiction, assume that there exists another major flow $H'_j$ that contains $H_j$ as a proper subgraph. As $H'_j$ is weakly connected due to (P2), there must be an edge $t = (f,b)$ in $H'_j$ but not in $H_j$, with at least one endpoint, say $f$, belongs to $V(H_j)$ (otherwise there won't be a path from vertices in $V(H_j)$ to vertices in $V(H'_j)\setminus V(H_j)$). Let $G_{i'}=(V_{i'},T_{i'})$ be the minimal major flow containing $t$, then $f \in V_{i'} \cap V(H_j) \neq \varnothing$, which means that $G_{i'}$ should belong to the connected component $C_j$. This implies that $t \in G_{i'} \subseteq H_j$, which contradicts our assumption about $t$. Thus, $H_j$ must be maximal. Finally, as the algorithm goes through all possible funder-beneficiary transactions in $S$, all major flows will be found.

We now discuss the complexity of the algorithm.
In Step~1, to construct $T_F(s)$ and $T_B(s)$ for all $s\in S$, the in- and out-transactions of $s$ must be sorted first, which dominates the running time of this step. Hence, Step~1 requires $O\big(|T(S)|\log(|T(S)|)\big)$ operations. In Step~2, each scam-funding transaction $t=(f,b)$, where $f,b\in S$, is examined at most twice, as an out-transaction of $f$ and an in-transaction of $b$ when expanding the sets $T_i$. Therefore, it takes $O(|T(S)|)$ operations to complete Step~2. In Step~3, the running time is dominated by the time used to build the auxiliary graph $\G$, which requires $O(|S||T(S)|^2)$ operations in the worst case using the most straightforward method with any optimization. Indeed, the algorithm needs to examine all pairs of $(G_i,G_j)$, $1\leq i < j \leq m$, where $m \leq |T(S)|$, and for each pair, the intersection $V_i \cap V_j$ can be found in $O(|S|)$ steps. 
\end{proof}

We observe that varying $p=80$-$95\%$ yields similar statistics for major flows in our datasets (see Table~\ref{tab:majority_flow_uniswap} and Table~\ref{tab:majority_flow_pancake}). 

\begin{table}[htb]
\setlength{\tabcolsep}{4pt}
\centering
\begin{tabular}{|r|r|l|l|l|l|}
\hline
\multicolumn{1}{|l|}{$p$} & \multicolumn{1}{l|}{\textbf{\#Flows}} & \textbf{Size} & \textbf{Width} & \textbf{Fund In} & \textbf{Fund Out} \\ \hline
95\%                             & 5,145                                    & 156;4      & 7;3       & 1,645;33     & 1,797;34      \\ \hline
90\%                             & 5,298                                    & 156;4      & 7;3       & 1,645;32     & 1,797;34      \\ \hline
85\%                             & 5,268                                    & 202;4      & 7;3       & 1,645;32     & 1,797;33      \\ \hline
80\%                             & 5,267                                    & 212;4      & 9;3       & 1,645;32     & 1,797;33      \\ \hline
\end{tabular}
\vspace{5pt}
\caption{\textmd{Statistics for \textit{maximal major flows} found in our
\textbf{Uniswap} scammer dataset for $p=80$-$95\%$. Rounded maximum and average values are reported.}}
\label{tab:majority_flow_uniswap}
\vspace{-10pt}
\end{table}

\begin{table}[htb]
\setlength{\tabcolsep}{4pt}
\begin{tabular}{|r|r|l|l|l|l|}
\hline
\multicolumn{1}{|l|}{$p$} & \multicolumn{1}{l|}{\textbf{\#Cluster}} & \textbf{Size} & \textbf{Width} & \textbf{Fund In} & \textbf{Fund Out} \\ \hline
95\%                             & 16,344                                   & 820;4      & 8;3       & 1,227;12     & 2,626;13      \\ \hline
90\%                             & 16,467                                   & 820;4      & 6;3       & 1,099;12     & 2,626;12      \\ \hline
85\%                             & 16,542                                   & 820;4      & 6;3       & 1,099;12     & 2,626;12      \\ \hline
80\%                             & 16,623                                   & 820;4      & 6;3       & 1,099;12     & 2,626;12      \\ \hline
\end{tabular}
\vspace{5pt}
\caption{\textmd{Statistics for \textit{maximal major flows} found in our
\textbf{Pancakeswap} scammer dataset for $p=80$-$95\%$. Rounded maximum and average values are reported.}}
\label{tab:majority_flow_pancake}
\end{table}

We also measured the AST-Jaccard similarity scores for token contracts in major flows (with $p=90\%$) that have at least two contracts available on Etherscan and BSCscan. On average, Uniswap flows have intra-flow similarity score of 84\% and inter-flow similarity score of 23\%. Moreover, 1,917 (51\%) of such flows have 100\% intra-flow similarity, and 2,694 (72\%) flows have intra-flow similarity scores at least 80\%. Also, 3,828 (93\%) of such flows have inter-flow similarity scores at most 30\%. On average, Pancakeswap flows have intra-flow similarity score of 77\% and inter-flow similarity score of 21\%. Moreover, 3,885 (30\%) of such flows have intra-flow similarity 100\%, and 8,372 (65\%) have intra-flow similarity scores at least 80\%. On the other hand, 13,928 (100\%) flows on Pancakeswap have inter-flow similarity scores at most 30\%.

\section{Scam Clusters: More Analysis}
\label{app:scam_clusters}

The top five clusters (ranked by the number of scammer addresses in them) with their total cluster-aware scam profits and the number of scam patterns founded in them are given in Table~\ref{tab:cluster_stats}.

\begin{table}[htb]
\begin{tabular}{|c|c|c|c|c|}
\hline
\textbf{ID} & \textbf{\#Scams} & \textbf{Size} & \textbf{\#Patterns} & \textbf{\begin{tabular}[c]{@{}c@{}}Profit\\ ETH/BNB\end{tabular}} \\ \hline
UNI-1               & 8,500            & 4,155                & 700                & 29,711      \\ \hline
UNI-2                & 1,602              & 646                  & 204                  & 16,307           \\ \hline
UNI-3               & 549              & 538                  & 166                  & 170          \\ \hline
UNI-4                & 433               & 432                   & 33                   & 456          \\ \hline
UNI-5                & 503               & 395                   & 53                  & 1,154          \\ \hline
CAKE-1             & 57,145                                & 23,399                                    & 4,500                                    & 5695                                                                                \\ \hline
CAKE-2             &1,330                                    & 1,327                                        & 447                                        & 461                                                                                 \\ \hline
CAKE-3             & 2,439                                    & 1,295                                        & 42                                        & 65                                                                                \\ \hline
CAKE-4             & 1,171                                    & 1,154                                        & 61                                        & 5                                                                                 \\ \hline
CAKE-5             & 3,550                                    & 1,135                                        & 153                                        & 48 \\ \hline
\end{tabular}
\vspace{5pt}
\caption{\textmd{The statistics of the five largest scam clusters on \textbf{UniswapV2} and \textbf{PancakeswapV2}. The ``Patterns'' column refers to the number of patterns (chains, stars, major flows) found in the cluster, whereas the ``Profit'' column refers to the total cluster-aware profit of all the pools in the cluster.}}
\label{tab:cluster_stats}
\vspace{-10pt}
\end{table}

The distributions of intra-cluster and inter-cluster similarities of Uniswap and Pancakeswap scam clusters are given in Figs.~\ref{fig:intra_uni},~\ref{fig:intra_pan},~\ref{fig:inter_uni}, and~\ref{fig:inter_pan}. In summary, scam contracts within the same clusters are highly similar, while scam contracts from different clusters are not similar, which strengthens our prediction that scammer addresses from the same cluster are likely possessed by the same scammer/organization, which tends to use more similar contracts, obtained potentially by cloning and obfuscating a few original ones. 

 \begin{figure}[htb!]
    \centering
    \includegraphics[width=0.45\textwidth]{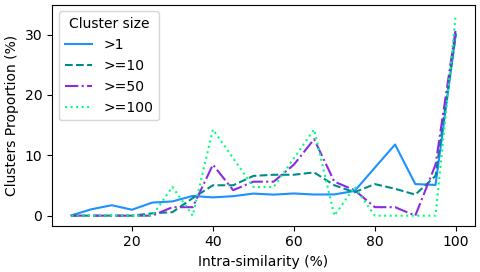}
    \caption{\textmd{\textit{Intra-cluster} similarities for Uniswap scam clusters.}}
    \label{fig:intra_uni}
    \vspace{-10pt}
\end{figure}

\begin{figure}[htb!]
    \centering
    \includegraphics[width=0.45\textwidth]{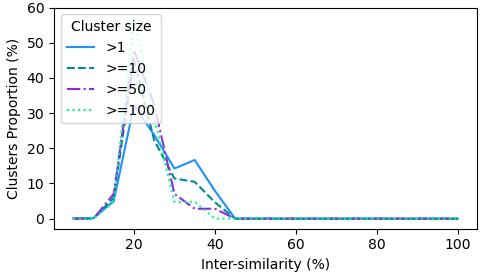}
    \caption{\textmd{\textit{Inter-cluster} similarities for Uniswap scam clusters.}}
    \label{fig:inter_uni}
    \vspace{-10pt}
\end{figure}



For scammer addresses that do not belong to any cluster, we manually inspected a random subset of them, which consists of addresses ending with `abc' or `000'. We found that most of them received funds from DeFi services including CEX, DEX, mixers, and bridges, such as Binance, Coinbase, Kucoin, FixedFloat, Bybit, MEXC, ChangeNOW, DLN, XY Finance, and transferred the funds to these services after scamming. Some also interacted with those services indirectly via intermediate, non-scammer addresses.

\section{Scam Networks \& True Scam Profits}
\label{app:scam_network}

\subsection{Multi-Hop Non-Scammer Wash Traders}
\label{app:multi_hop_non_scammer_WT}

\rev{We first discuss our observation of wash trader addresses that are non-scammer and multiple hops away from the scammer address.

\begin{figure}[t]
    \centering
    \includegraphics[scale=1]{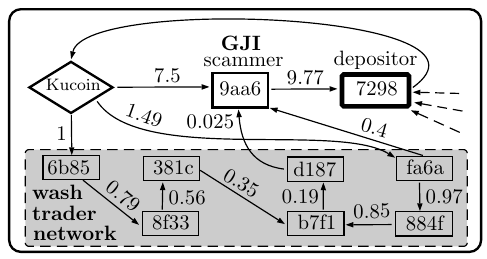}
    \caption[abc]{\textmd{The wash traders of the scam token \texttt{GJI} on Uniswap created by the scammer address (ending with) \textbf{9aa6} (its full address is 0x19b98792e98c54F58C705CDDf74316aEc0999AA6) never received fund directly from it. Several are multiple hops away and funded by a CEX (KuCoin) or among themselves. None of the wash traders are scammer addresses.}}
    \label{fig:wt_9aa6}
\end{figure}

Note that \textit{1-hop} wash-trader addresses was observed in Xia~\et~\cite{Xia_etal_2021}, which received direct ETH transfers from scammer addresses before swapping scam tokens. 
A recently developed wash-trading detection tool by SolidusLab~\cite{soliduslabs_washtrading_blog, soliduslabs_washtrading_tool} can only detect \textit{0-hop} wash traders (self wash trading).
However, we noticed in our datasets wash traders that were multiple hops away from the scammer addresses and never interacted directly with them. We illustrate one such case in Example~\ref{ex:GJI}.} 

\begin{example}
    \label{ex:GJI}
    The scammer address (ending with) \textbf{9aa6}, after receiving 7.5 ETH from the CEX KuCoin, created the scam token \texttt{GJI}, and added 7 ETH and $10^{12}$ \texttt{GJI} as liquidity to a scam pool, which was subsequently wash traded by a small but tightly-connected \textit{network} of seven wash-trader addresses, the furthest five hops away from \textbf{9aa6} (see Fig.~\ref{fig:wt_9aa6}). Note that as a public CEX, Kucoin should not be considered when counting hops between addresses. Other than wash trading for \texttt{GJI} and transferring among themselves and the scammer address, these addresses didn't do anything else, which means their sole purpose was wash trading. 
    
    The \textit{wash-trading network} for \texttt{GJI} worked as follows. First, \textbf{6b85} received 1 ETH from Kucoin, then swapped 0.2 ETH for \text{GJI} a minute after the liquidity was added by \textbf{9aa6}. It also transferred 0.79 ETH to the next wash trader \textbf{8f33}, which bought \texttt{GJI} twice using a total of 0.2 ETH. \textbf{8f33} then transferred 0.56 ETH to \textbf{381c}, which swapped 0.2 ETH for \texttt{GJI}, before transferring 0.35 ETH to \textbf{b7f1}. The address \textbf{b7f1} also received 0.85 ETH from another wash trader \textbf{884f} of \texttt{GJI}, which was funded by yet another wash trader \textbf{fa6a}, which was originally funded from Kucoin. \textbf{bf71} bought \texttt{GJI} with 1 ETH, and transferred 0.19 ETH to \textbf{d187}, which, after swapping 0.15 ETH for \texttt{GJI}, transferred the leftover fund of 0.025 ETH to the scammer address \textbf{9aa6}. The address \textbf{fa6a} also transferred 0.4 ETH to \textbf{9aa6}. Finally, \textbf{9aa6} removed its liquidity with 9.29 ETH from the scam pool and transferred 9.77 ETH to a \textit{depositor} address \textbf{7298}, which received funds from more than a hundred of (similar) scammer and transferrer addresses and deposited to Kucoin also more than a hundred times.
\end{example}

\vspace{-5pt}
\subsection{Node Labelling}
\label{subsec:labeling}

We first define the main roles performed by addresses in a scam network, which are associated with three key operations of such networks: \textit{scamming} (Rug Pull), \textit{wash trading} (for scam pools), and (scam) \textit{money laundering}. All addresses must be in the network.
\begin{itemize}
    \item $\ts$: \textit{scammer} address - associated with a one-day scam Rug-Pull pool (see Definition~\ref{def:scammer_address}).
    \item $\tc$: \textit{coordinator} address - was the \textit{largest funder} of at least five scammer addresses, and at least 50\% of its EOA neighbors must be scammer addresses. The largest funder of an address A is one of the in-neighbours that transferred the maximum amount of high-value token (ETH/BNB) to A. 
    \item $\twt$: \textit{wash-trader} address - reachable from a scammer address within the network after a series of native-token transfers (the directions of such transfers are not important), and bought at least one one-day scam token from a scam pool within that network. 
    \item $\td$/$\tw$: \textit{depositor/withdrawer} address - sent fund to/withdraw from CEXs, mixers, or bridges.
    \item $\tt$: \textit{transfer} address - received and forwarded fund only and performed no other activities nor retained the fund received. More specifically, the address must only interact with EOAs and moreover, the total out transfers (including transaction fees) must be at least 99\% of the total in transfers (not including transaction fees).
    \item $\tb$: \textit{boundary} address - is not scammer or coordinator, has at least ten token swap-ins and more than 50\% of the swap-ins were with non-scam pools.
\end{itemize}

Note that to identify wash trader addresses of a scam network, we only examine native-token transfers and believe that most wash-trader addresses should be discovered that way.
We argue that it is unlikely that a wash trader address receives the scam tokens from another address within the scam network to perform wash trading. First, it may raise immediate suspicion as a normal investor would never own a (scam) token without buying first. Second, the wash traders should buy the scam tokens to increase their prices to make them more attractive, and not to sell right away to reduce the prices. And third, it is a less straightforward way to manage the wash trading network and potentially more expensive compared to direct native-token transfer. Thus, most rational scammers wouldn't transfer scam tokens to their wash-trader addresses. In fact, we examined all scam pools, and found that only 3,368 pools (2.1\%) on Uniswap and 11,059 pools (2.3\%) on Pancakeswap have investor addresses that received scam tokens from other addresses. 

\subsection{Generating Scam Network}
\label{subsec:scam_network}


A natural choice to explore the scam network is to use a Breadth-First Search (BFS) as it can capture essential chain activities including fund transfer/deposit/withdraw and DEX-specific activities such as token swap and transfer, token/pool creation, liquidity providing and removal. Although theoretically straightforward, implementing BFS to identify scam networks on Ethereum and BSC chains is remarkably difficult. There are four main challenges, but the major one is to avoid network explosion and to prevent the BFS from including public or benign addresses.  
We will use accounts, addresses, and nodes interchangeably.


\textbf{Challenge 1.}  \textit{The background graph is non-homogeneous}. Nodes in the transaction graph include externally owned accounts (EOAs) and contract accounts. Contract accounts could be mixers, centralized exchanges, bridges, routers, MEV/trading bots, token deployers, each type behaves differently and requires a different treatment.

\textbf{Challenge 2.} \textit{Identifying an edge between two accounts is challenging and requires domain knowledge}. For example, an edge exists between two EOAs A and B not only when A transferred fund to B directly, but also when A and B jointly created a scam. On the other hand, if A sends a 0-value transaction to B (a message-carrying transaction), that shouldn't be counted as an edge\footnote{Our algorithm initially hit a non-scam account that sent a transaction to a known scam account. It turns out that the transaction only carries a sneering message.}.

\textbf{Challenge 3.} \textit{The background graph is huge and unavailable}. By contrast to the standard setting in the traditional graph theory, due to its sheer size, the complete underlying transaction graph among all accounts on Ethereum or BSC (or any other established chain) is not available and impossible to build\footnote{For example, the Ethereum chain generates more than a million new transactions everyday (see \url{https://ycharts.com/indicators/ethereum_transactions_per_day}).}. 

\textbf{Challenge 4.} \textit{Network explosion}. Starting from a (seed) scammer node, the standard BFS can quickly expand to include an unmanageable number of nodes in its queue, especially when the node in consideration starts transacting with public accounts (CEXs, mixers, bridges), or trading from non-scam pools. Thus, we must be able to recognize the boundary nodes associated with non-scam activities, to prevent BFS from including normal nodes unrelated to the scammers, which may have been active for years with thousands to hundreds of thousands of transactions and would fairly quickly overwhelm BFS's memory and pollute the real scam network. 

Apart from the above, irrelevant transactions generated by \textit{phishing attacks} (e.g. address-poisoning attacks) can also confuse the BFS, making it jump out of the scam network by mistake.

\textbf{Address-poisoning attacks can break BFS.} Phishing attackers target everyone, including addresses in Rug-Pull scam networks. Dusting attack or dust value transfer  (see, e.g.~\cite{Binance_20M_phishing,GuanLi_CCS_2024, Chen_etal_NDSS2025}) is a common type of address poisoning attack in which an address that looks very similar to the victim's out-neighbor (see Fig.~\ref{fig:address_poisoning}) transferred a tiny amount to the victim address. 
A naive BFS may expand from a node in the scam network that was the target of a dusting attack to visit the attacker address, then to its coordinator address (the one who coordinates all the phishing attacks) and then unknowingly to \textit{all} of its (benign) victim addresses in the network. We implemented a simple phishing detector that excludes a phishing address from the set of valid neighbors for each address in BFS, which detected dozens to hundreds of such addresses for every network we tried.

\begin{figure}[htb]
    \centering
    \includegraphics[scale=0.39]{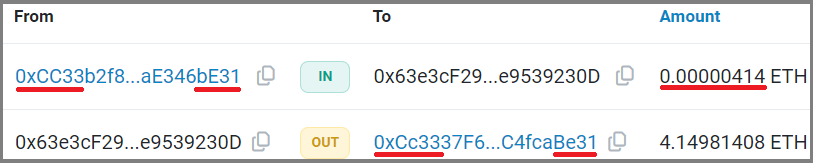}
    \caption{\textmd{The phishing address has the same (lower-cased) first and last four digits as an out-neighbor of the victim address. It transferred a tiny amount to the address \textbf{230d}, hoping that one day the address owner will mistakenly transfer fund to it instead.}}
    \label{fig:address_poisoning}
\end{figure}

\textbf{Existing approaches.} To circumvent the network explosion problem (Challenge 4) when building the transaction subgraph (starting from one or more seed nodes), one must identify/define a boundary, or \textit{terminal} nodes, at which the BFS stops expanding further. The simplest way is to set the boundary nodes to be just \textit{1-hop neighbors} of the scammer nodes (equivalently, from the scam pools). For example, Morgia~\et~\cite{Morgia_etal_ICDCS_2023} investigated weakly connected components of the subgraph generated by the 1-hop neighbours of the NFT scams to identify wash-trading groups. The issue with this approach is that it treats each scam as an isolated one and fails to recognize the connection among multiple scams (the main topic of our work). Another simple way is to explicitly set the \textit{maximum number of hops} the BFS can reach, e.g. ten hops as in Yan~\et~\cite[Algo.~1]{Yan_etal_Cybersecurity_2023}. 
However, this artificial threshold (ten hops) would lead to an inaccurate picture of a true scam clusters. For instance, we discovered in our work a number of very long scam chains with lengths up to a few hundreds (see Section~\ref{subsec:chain}). Limiting the BFS to a fixed, small number of hops would also lead to inaccurate statistics on scam clusters and their true profits.

\textbf{Our modified BFS} receives as input the scammer list $S$, the scam pool list $P$, and scam token list $T$ identified in the data collection phase. It then iterates over $S$, starts from each unvisited address as a seed, retrieves the transaction history of the current address in consideration, and identifies and places its \textit{valid} neighbors into BFS's queue for future processing.
Valid neighbors of the current address $v$ include unvisited/unqueued EOAs that had a \textit{non-zero} native-token (ETH/BNB) transfer with $v$ or were scammer addresses behind the same scam pool as $v$ (if $v$ is a scammer). Note that the BFS must identify address-poisoning transactions to avoid including the attacker/other victim addresses into the scam network. 

\textbf{To address network explosion} (Challenge 4), instead of setting the maximum number of hops like in~\cite{Morgia_etal_ICDCS_2023,Yan_etal_Cybersecurity_2023}, we allow BFS to expand arbitrarily far, but identifying \textit{terminal} nodes at which BFS stops expanding. The set of terminate nodes is described as follows.
\textit{Public terminal nodes} are publicly label nodes such as mixers, CEXs, bridges, MEVs, contract deployers, and DEX routers. Similar to most work in the literature, a list of such addresses can be collected from well-known sites such as Etherscan's WordCloud and Dune, and also manually added on the fly. \textit{Normal trading (boundary) nodes} are non-scammer/coordinator addresses that had at least 10 swap-ins and at least 50\% of them were with non-scam exchange pool. If such an address is reached, BFS won't add its neighbors to the queue.
\textit{Big nodes} are defined according to two limits $\ell=500$ and $L=1000$. Addresses that have more than $L$ transaction will be ignored. Addresses with more than $\ell$ but at most $L$ transactions will be ignored except for scammers and coordinators.


\begin{figure*}[htb]
    \centering
    \includegraphics[scale=0.45]{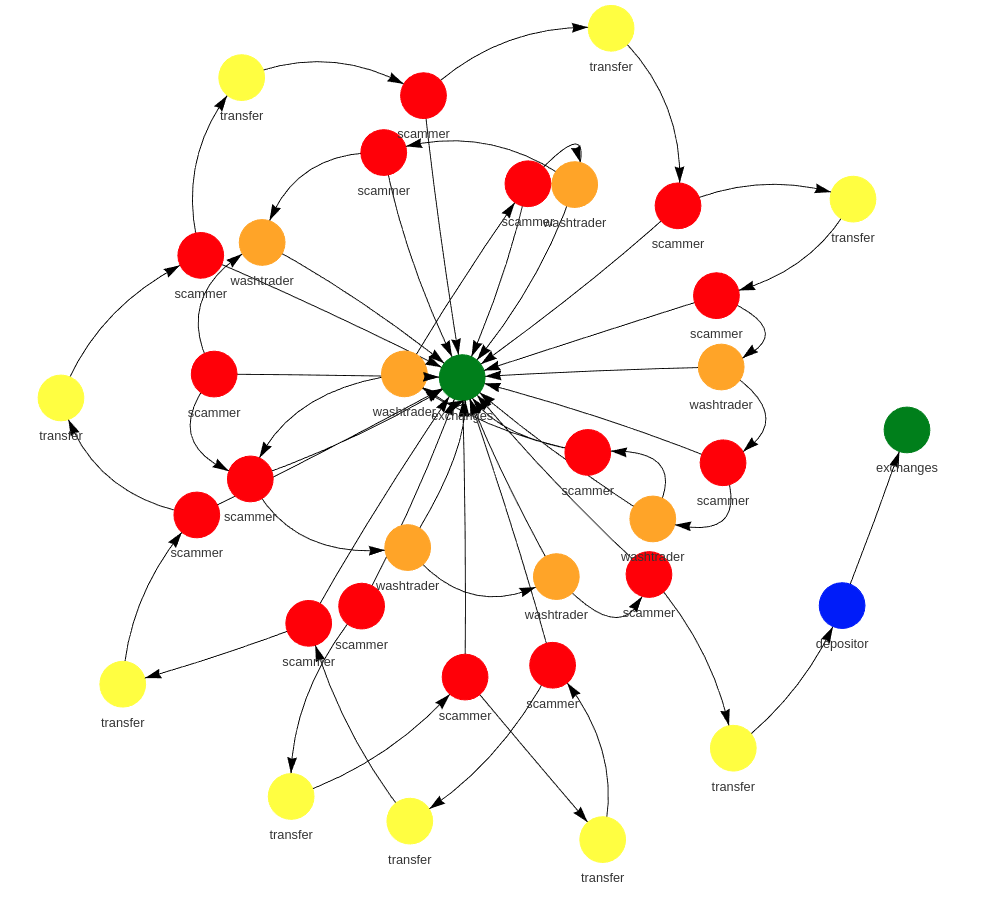}
    \caption{\textmd{The scam network on Uniswap expanded from the cluster 85. Color code: red - scammer, orange - wash trader, yellow - transfer, blue - depositor, purple - coordinator, green - (public) exchange. The arrows represent the native-token transfers.}}
    \label{fig:Uniswap_sample_network}
\end{figure*}

\begin{figure*}[htb]
    \centering
    \includegraphics[scale=0.45]{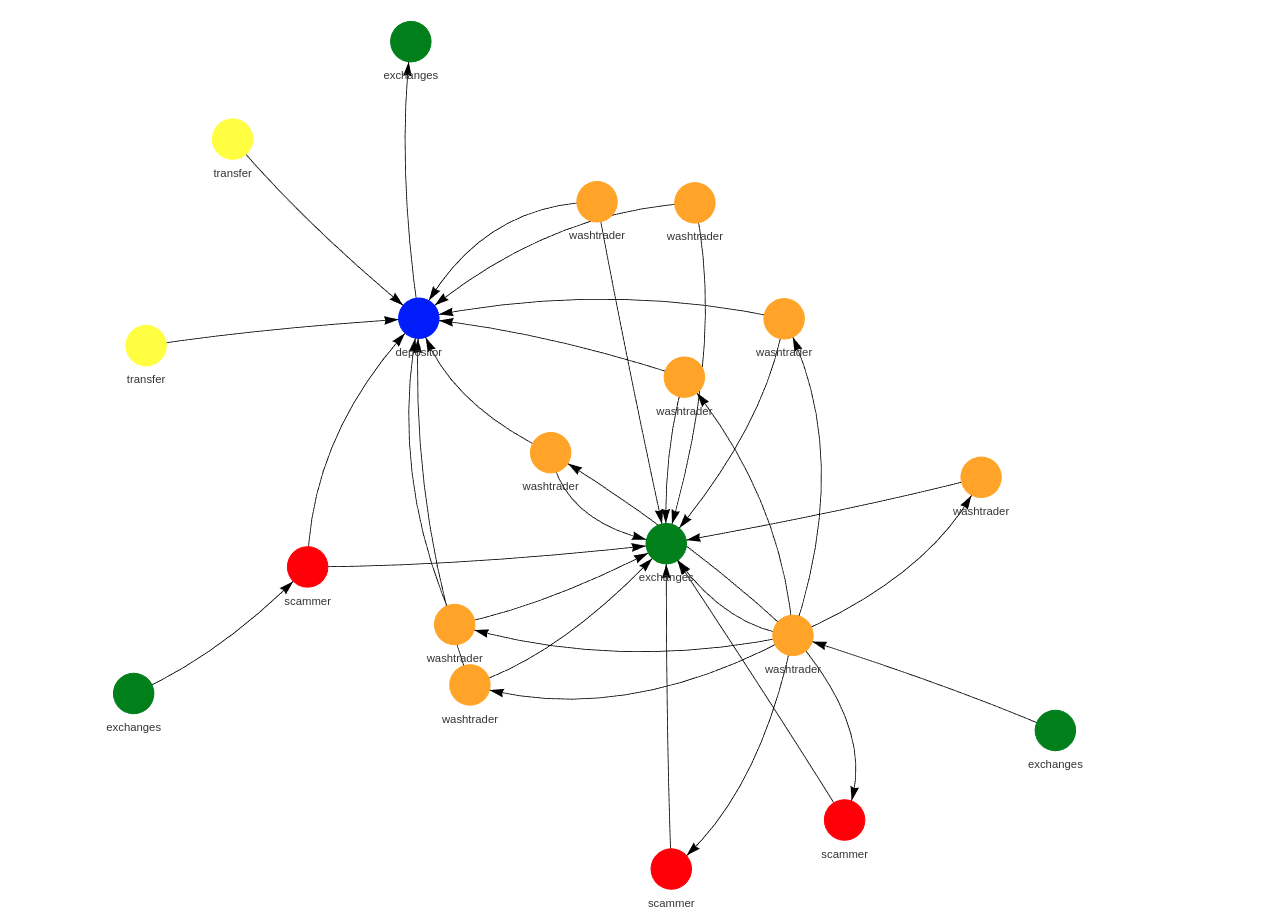}
    \caption{\textmd{The scam network on Pancakeswap expanded from the cluster 8504.}}
    \label{fig:Pancakeswap_sample_network}
\end{figure*}

\textbf{Small scam networks.} We ran our BFS algorithm (see App.~\ref{subsec:scam_network}) starting from the scam clusters as seed, limited the network size to 200 nodes, and obtained 134 Uniswap and 35 Pancakeswap small but complete sample scam networks. These networks contain in total 2,751 Uniswap scam pools and 1,014 Pancakeswap scam pools. Moreover, the Uniswap networks contain 6,252 scammer addresses and 2,680 non-scammer addresses. The Pancakeswap networks contain 1,311 scammer addresses and 897 non-scammer addresses. Lastly, 39 out of 134 Uniswap networks have wash traders, for which the averaged total network scam profit produced by the naive formula is 58\% higher than that from the \textit{network-aware} formula. There are 17 out of 35 Pancakeswap networks with wash traders in them, for which the averaged total network scam profit generated by the naive formula is 115\% higher than what is obtained using the network-aware formula. Two sample scam networks with nodes labeled by their roles are given in Fig.~\ref{fig:Uniswap_sample_network} and Fig.~\ref{fig:Pancakeswap_sample_network}.

\end{document}